\documentclass[journal]{IEEEtran}

\usepackage{cite}

\ifCLASSINFOpdf
  \usepackage[pdftex]{graphicx}
  \DeclareGraphicsExtensions{.pdf,.jpeg,.png}
\else
\fi
\usepackage[cmex10]{amsmath}
\usepackage{algorithm}
\usepackage{algpseudocode}

\usepackage{array}

\ifCLASSOPTIONcompsoc
  \usepackage[caption=false,font=normalsize,labelfont=sf,textfont=sf]{subfig}
\else
  \usepackage[caption=false,font=footnotesize]{subfig}
\fi
\usepackage{amsthm}
\newtheorem{theorem}{Theorem}
\newtheorem{lemma}{Lemma}

\newtheorem{corollary}{Corollary}

\theoremstyle{definition} \newtheorem{definition}{Definition}

\newtheorem{assumption}{Assumption}
\newtheorem{remark}{Remark}
\newtheorem{desired_feature}{Desired Feature}
\newtheorem{requirement}{Requirement}

\usepackage{color}

\usepackage[normalem]{ulem}

\usepackage{amsfonts}
\usepackage{amssymb}

\begin{document}
%
\title{
Bio-Inspired Local Information-Based Control for Probabilistic Swarm Distribution Guidance 
}
%
%
%

\author{Inmo~Jang,~
        Hyo-Sang~Shin,~
        and~Antonios~Tsourdos
\thanks{Inmo Jang, Hyo-Sang Shin, and Antonios Tsourdos are with the Centre for Autonomous and Cyber-Physical Systems, 
Cranfield University, MK43 0AL, United Kingdom
(e-mail: inmo.jang@cranfield.ac.uk; h.shin@cranfield.ac.uk; a.tsourdos@cranfield.ac.uk). }
}

\maketitle

\begin{abstract} 
This paper addresses a task allocation problem for a large-scale robotic swarm, namely \emph{swarm distribution guidance problem}. Unlike most of the existing frameworks handling this problem, the proposed framework suggests utilising \emph{local} information available to generate its time-varying stochastic policies. 
As each agent requires only local consistency on information with neighbouring agents, rather than the global consistency, the proposed framework offers various advantages, e.g., a shorter timescale for using new information and potential to incorporate an asynchronous decision-making process. 
We perform theoretical analysis on the properties of the proposed framework. 
From the analysis, it is proved that the framework can guarantee the convergence to the desired density distribution even using local information while maintaining advantages of global-information-based approaches. 
The design requirements for these advantages are explicitly listed in this paper. 
This paper also provides specific examples of how to implement the framework developed. 
The results of numerical experiments confirm the effectiveness and comparability of the proposed framework, compared with the global-information-based framework. 
\end{abstract}


%
\IEEEpeerreviewmaketitle




\section{Introduction}

\IEEEPARstart{T}{his} paper addresses a task allocation problem for a large-scale multiple-robot system, called a \emph{robotic swarm}. 
Robotic swarms have attracted lots of attention because they are regarded as promising solutions to handle complicated missions that other systems may not be able to manage \cite{Brambilla2013, Kolling2016}. 
Agents in a swarm are assumed to be homogeneous because the swarm is usually realised through mass production \cite{Sahin2005}.  
In this context, the task allocation problem can be reduced to a problem of how to distribute a swarm of agents into given tasks (or bins), satisfying the desired population fraction (or swarm density) for each task. 
This problem is known as the \emph{swarm distribution guidance problem} \cite{Acikmese2012, Acikmese2014, Acikmese2015}. 

For a large number of agents,  
probabilistic approaches based on Markov chains \cite{Chattopadhyay2009, Acikmese2012, Acikmese2014, Acikmese2015, Demir2015, Luo2014, Demir2015a, Bandyopadhyay2013a, Morgan2014, Bandyopadhyay2017} or 
differential equations \cite{Halasz2007, Hsieh2008, Berman2008, Berman2009, Mather2011} 
have been widely utilised. 
Since these approaches focus not on individual agents but instead on the ensemble dynamics, 
they are also called \emph{Eulerian} \cite{Bandyopadhyay2013a, Morgan2014, Bandyopadhyay2017} or 
\emph{macroscopic} frameworks \cite{Lerman2005, Mather2011}.
In these approaches, swarm densities for each bin are represented as system states, and 
a state-transition matrix describes \emph{stochastic decision policies}, i.e., the probabilities that agents in a bin switch to another. 
Individual agents in the swarm make decisions based on these policies, but in a random, independent, and memoryless manner. 

Initially, \emph{open-loop-type} frameworks have been proposed  \cite{Berman2008, Berman2009, Mather2011, Chattopadhyay2009, Acikmese2012, Acikmese2014, Acikmese2015}. Agents under these frameworks are controlled by time-invariant stochastic decision policies. 
The policies, which make a swarm converge to a desired distribution, are pre-determined by a central controller and broadcasted to each agent before executing the mission. 
Communication between agents is hardly required during the mission, so that it can reduce communication complexity under these frameworks.  
However, the agents only have to follow the given policies without incorporating any feedbacks,  
and thus there still remain some agents who unnecessarily and continuously switch bins even after the swarm reaches the desired distribution. 
This gives rise to a trade-off between convergence rate and long-term system efficiency \cite{Berman2009}.
There have been also 
some other works, called \emph{closed-loop-type} frameworks \cite{Halasz2007, Hsieh2008, Luo2014, Demir2015a, Bandyopadhyay2013a, Morgan2014, Bandyopadhyay2017}. This type of frameworks allows agents to adaptively construct their own stochastic decision policies at the expense of sensing the concurrent swarm status through interactions with other agents. 
Based on such information, agents can synthesise time-inhomogeneous transition matrices to achieve certain objectives and requirements: for example, maximising convergence rates \cite{Demir2015a}, minimising travelling costs \cite{Bandyopadhyay2017}, and temporarily adjusting given policies when bins are more overpopulated or underpopulated than certain levels \cite{Halasz2007, Hsieh2008}. 
In particular, Bandyopadhyay et. al. \cite{Bandyopadhyay2017} recently proposed a closed-loop-type algorithm that 
exhibits faster convergence as well as less undesirable transition behaviours, compared with an open-loop-type algorithm. 
This algorithm is expected to mitigate the trade-off raised in open-loop-type frameworks. 

To the best of our knowledge, most of the existing closed-loop-type algorithms are based on \emph{Global Information Consistency Assumption} (GICA) \cite{Johnson2016}.
GICA implies that necessary information is required to be consistently known by entire agents. 
We refer to such information as \emph{global}, because 
achieving information consistency needs agents to somehow interact with all the others through a multi-hop fashion and thus it ``happens on a global communication timescale''\cite{Johnson2016}. 

This paper proposes a framework that requires \emph{Local Information Consistency Assumption} (LICA) \cite{Johnson2016}. Unlike GICA-based algorithms, the proposed framework require only local consistency on information with neighbouring agents, not the global consistency. 
LICA can provide various alternative advantages to the proposed framework, compared with GICA.  
Firstly, it ``provides a much shorter timescale for using new information because agents are not required to ensure that this information has propagated to the entire team before using it''\cite{Johnson2016}.
%
Secondly, LICA enables a foundation on which an asynchronous decentralised decision-making process can be developed. Note that the timescales for achieving the information consistency between the agents can be different depending on their local circumstances.  
Considering any possibly-extrinsic heterogeneity of agents (e.g., different sensing frequency due to local communication delays), an asynchronous algorithm is regarded as more realistic in coordinating a robotic swarm, so increasing its system efficiency \cite{Xue2014, Koh2006, Johnson2011}.
%
Finally, LICA makes the proposed approach additionally robust against dynamical changes in bins and those in agents.
Given that inclusions or exclusions of bins are perceived by neighbouring agents, the proposed approach works well even without requiring other far-away agents to know the changes.  

The LICA-based framework developed in this paper utilises \emph{local} information 
as its feedback gains, which is motivated from the recent GICA-based work in \cite{Bandyopadhyay2017}. 
This framework is inspired by the mechanism of decision-making in a fish swarm, in which each of them adjusts its individual behaviour based on those of neighbours \cite{Couzin2002, Couzin2005, Gautrais2008, Hoare2004}. 
Similarly, each agent in the framework developed uses its local status, i.e. the current density of its associated bin relative to those of its neighbour bins, to generate its time-varying stochastic decision policies.
The agent is not required to know any global information, and hence the aforementioned advantages of LICA can be exploited. 

We prove that, even using local information, the proposed framework asymptotically converges to a desired swarm distribution and it retains the advantages of existing closed-loop-type approaches.
This paper explicitly presents the design requirements for a time-inhomogeneous Markov chain to achieve these desired features. It is thus expected that the user can utilise the requirements in designing  their own algorithm.  
In addition, three specific examples are provided to demonstrate how to implement the proposed framework: 1) minimising travelling cost; 2) maximising convergence rate under upper flux bounds; and 3) generation of quorum-based policies (similar to \cite{Halasz2007, Hsieh2008}).

The rest of this paper are organised as follows. 
Section \ref{sec:preliminary} introduces the desired features of a swarm distribution guidance framework along with relevant definitions and notations. 
Section \ref{sec:framework0} proposes our framework with its design requirements, the biological inspiration, and an analysis regarding whether the desired features are satisfied. 
We provide examples of how to implement the framework for specific problems in Section \ref{sec:implementation_example}, and an asynchronous implementation in Section \ref{sec:async}. 
The results of numerical experiments are shown in Section \ref{sec:experiment}, followed by concluding remarks in Section \ref{Conclusion}.

\subsection*{Notations}
$\emptyset$, \textbf{0}, $I$ and \textbf{1} denote the empty set, the zero matrix of appropriate sizes, the identity matrix of appropriate sizes, and a row vector with all elements are equal to one, respectively.
$v \in \mathbb{P}^{n}$ is a stochastic (row) vector such that $v \ge \textbf{0}$ and $v \cdot \textbf{1}^{\top} = 1$.
$v[i]$ indicates the $i$-th element of vector $v$.
$\mathrm{Prob}(E)$ denotes the probability that event $E$ will happen.

\begin{table}[h]
\renewcommand{\arraystretch}{1.3}
\caption{Nomenclature}
\label{nomenclature}
\centering
\begin{tabular}{p{0.4in} p{2.7in}}
\hline
Symbol & Description \\
\hline
\hline
$\mathcal{B}_i$	 		& The $i$-th bin amongst a set of $n_{bin}$ bins (Definition \ref{def:bin}); \\
$\mathcal{A}$	& A set of agents  (Definition \ref{def:bin}); \\
$n^{\mathcal{A}}_k$	& The number of total agents at time instant $k$ (Definition \ref{def:bin}); \\
$n_k[i]$	& The number of agents at the $i$-th bin (Eqn. (\ref{eqn:local_dist})); \\
$a^j_k$		& The $j$-th agent's state indicator vector (Definition \ref{def:state});\\
$x^j_k$	& Stochastic state vector of the $j$-th agent (Definition \ref{def:prob_state}); \\
$M^j_k$	& Stochastic decision policy of the $j$-th agent (Definition \ref{def:prob_state}); \\
$\Theta$ 	& Desired swarm distribution (Definition \ref{def:theta});\\
$\mu_k^{\star}$ & Current (global) swarm distribution (Definition \ref{def:current_distribution});\\
$A_k$ & Physical motion constraint matrix (Definition \ref{def:physical_route});\\
$C_k$ & Communicational connectivity matrix (Definition \ref{def:comm_connect});\\
$\mathcal{N}_k(i)$ & A set of (communicationally-connected) neighbour bins of the $i$-th bin (Definition \ref{def:comm_connect});\\
$\mathcal{A}_{\mathcal{N}_k(i)}$ & A set of agents in $\mathcal{N}_k(i)$;\\
$\bar{\mu}_k^{\star}[i]$ & Current \emph{local} swarm density at the $i$-th bin (Eqn. (\ref{eqn:local_dist}));\\
$\bar{\mu}_k^{j}[i]$ & Estimate of $\bar{\mu}_k^{\star}[i]$ by the $j$-th agent;\\
$\bar{\Theta}[i]$ & Locally-desired swarm density at the $i$-th bin (Eqn. (\ref{eqn:local_dist_desired}));\\
$P^j_k$   & Primary guidance matrix (Eqn. (\ref{eqn:P_matrix}));\\
$S^j_k$ 	& Secondary guidance matrix (Eqn. (\ref{eqn:P_matrix}));\\
$\bar{\xi}^j_k[i]$ 	& Primary local-feedback gain (e.g., Eqn. (\ref{eqn:xi}));\\
$G^j_k[i]$ & Secondary local-feedback gain (Eqn. (\ref{eqn:eta}));\\
\hline
\end{tabular}
\end{table}

\section{Preliminaries}\label{sec:preliminary}
 
\subsection{Definitions}

This section presents necessary definitions and assumptions for our proposed framework, which will be shown in Section \ref{sec:framework0}. 
Since most of them are embraced from the recent existing literature \cite{Demir2015a, Bandyopadhyay2017}, we here briefly provide their essential meanings. 

\begin{definition}[\emph{Agents and Bins}]\label{def:bin}
A set of agents $\mathcal{A}$ are supposed to be distributed over a prescribed region in a state space $\mathcal{B}$.  
The entire space is partitioned into $n_{bin}$ disjoint \emph{bins} (subspaces) such that 
$\mathcal{B} = \cup_{i=1}^{n_{bin}} \mathcal{B}_i$ and $\mathcal{B}_i \cap \mathcal{B}_j = \emptyset$, $\forall i \neq j$.
We also regard $\mathcal{B} = \{\mathcal{B}_1,...,\mathcal{B}_{n_{bin}}\}$ as the set of all the bins. 
Each bin $\mathcal{B}_i$ represents a predefined range of an agent's state, e.g., position.
The number of the entire agents is time-varying, and its value at time instant $k$ is denoted by $n^{\mathcal{A}}_k = |\mathcal{A}|$. 
Note that we do not assume that the agents keep track of $n^{\mathcal{A}}_k$. 
\end{definition}

\begin{definition}[\emph{Agent's state}]\label{def:state}
Let $a^j_k \in \{0,1\}^{n_{bin}}$ be the state indicator vector of agent $j \in \mathcal{A}$ at time instant $k$.
If the agent's state belongs to bin $\mathcal{B}_i$, then $a^j_k[i] = 1$, otherwise $0$. 
\end{definition}

\begin{definition}[\emph{Current (global) swarm distribution}]\label{def:current_distribution}
\emph{The current (global) swarm distribution} $\mu_k^{\star} \in \mathbb{P}^{n_{bin}}$ is a row-stochastic vector such that 
each element $\mu_k^{\star}[i]$ is the population fraction (swarm density) of $\mathcal{A}$ in bin $\mathcal{B}_i$ at time instant $k$:
\begin{equation}\label{eqn:current_distribution}
\mu_k^{\star} := \frac{1}{|\mathcal{A}|} \sum_{\forall j \in \mathcal{A}} a^j_k.
\end{equation}
 %
\end{definition}


\begin{definition}[\emph{Agent's stochastic state and decision policy}]\label{def:prob_state}
Agent $j$'s \emph{stochastic state} is a row-stochastic vector $x^j_k \in \mathbb{P}^{n_{bin}}$  
in which each element $x^j_k[i]$ gives the probability that the agent's state belongs to bin $\mathcal{B}_i$ at time instant $k$:
\begin{equation}
x^j_k[i] := \mathrm{Prob}(a^j_k[i] = 1).
\end{equation} 
The probability that agent $j$ in bin $\mathcal{B}_i$ at time instant $k$ will transition to bin $\mathcal{B}_l$ before the next time instant is called its \emph{stochastic decision policy}, denoted as:
\begin{equation}\label{eqn:stochastic_policy}
M^j_k[i,l] := \mathrm{Prob}(a^j_{k+1}[l] = 1| a^j_k[i] = 1).
\end{equation}
Note that $M^j_k \in \mathbb{P}^{n_{bin} \times n_{bin}}$ is a row-stochastic matrix such that $M^j_k \ge \textbf{0}$ and $M^j_k \cdot \textbf{1}^{\top} = \textbf{1}^{\top}$, and will be referred as \emph{Markov matrix}.  
\end{definition}

\begin{definition}[\emph{Desired swarm distribution}]\label{def:theta}
\emph{The desired swarm distribution} $\Theta \in \mathbb{P}^{n_{bin}}$ is a row-stochastic vector such that
each element $\Theta[i]$ indicates the desired swarm density for bin $\mathcal{B}_i$. 
\end{definition}

\begin{assumption}
For ease of description for this paper, we assume that $\Theta[i] > 0$, $\forall i \in \{1,...,n_{bin}\}$.  
Obviously, in practice, there may exist some bins whose desired swarm densities are zero. 
These bins can be accommodated by adopting any subroutines ensuring that all agents eventually move to and remain in any of the positive-desired-density bins, for example, an escaping algorithm in \cite[Section III.C]{Bandyopadhyay2017}. 
\end{assumption}

\begin{assumption}[\emph{The number of agents} \cite{Hsieh2008, Berman2009, Acikmese2015, Demir2015a, Bandyopadhyay2017}]
It is assumed that $n^{\mathcal{A}}_k \gg n_{bin}$ so that the time evolution of the swarm distribution is governed by the stochastic decision policy in Equation (\ref{eqn:stochastic_policy}). 
Although the finite cardinality of the agents normally cause a residual convergence error, 
a lower bound on $n^{\mathcal{A}}_k$ that probabilistically guarantees a desired convergence error is analysed in \cite[Theorem 6]{Bandyopadhyay2017} by exploiting Chebyshev's equality. 
Note that this theorem is generally appliable and thus is also valid for our work. 
\end{assumption}

\begin{definition}[\emph{Physical motion constraint}  \cite{Acikmese2015, Demir2015a, Bandyopadhyay2017}]\label{def:physical_route}
Motion constraints of agents are denoted by the matrix $A_k \in \{0,1\}^{n_{bin} \times n_{bin}}$, 
where $A_k[i,l] = 1$ if agents in bin $\mathcal{B}_i$ at time instant $k$ are allowed to transition to bin $\mathcal{B}_l$ by the next time instant; 
$A_k[i,l] = 0$, otherwise. 
It is assumed that $A_k$ is symmetric and irreducible (i.e., strongly-connected); and $A_k[i,i] = 1$ for all agents, bins, and time instants. 
\end{definition}

\begin{definition}[\emph{Communicationally-connected}]\label{def:comm_connect}
Bins $\mathcal{B}_i$ and $\mathcal{B}_l$ are said to be \emph{communicationally-connected}, 
if 
there exists at least one agent in bin $\mathcal{B}_i$ who can directly communicate with some agents in bin $\mathcal{B}_l$, and vice versa. 
This communicational connectivity over all the bins at time instant $k$ is defined by the matrix $C_k \in \{0,1\}^{n_{bin} \times n_{bin}}$, 
where $C_k [i,l] = 1$ indicates that bins $\mathcal{B}_i$ and $\mathcal{B}_l$ are communicationally-connected. 
Note that $C_k$ is symmetric and all its diagonal entries are set to be one. 
For each bin $\mathcal{B}_i$, we define the set of its (communicationally-connected) \emph{neighbour bins} as $\mathcal{N}_k(i) = \{ \forall \mathcal{B}_l \in \mathcal{B} \ | \ C_k[i,l] = 1 \}$.
The set of agents in any of bins in $\mathcal{N}_k(i)$ is denoted by $\mathcal{A}_{\mathcal{N}_k(i)} = \{\forall j \in \mathcal{A} \ | \ a^j_k[l] = 1, \ \forall l:\mathcal{B}_l \in \mathcal{N}_k(i)\}$.  
\end{definition}

\begin{assumption}[{\emph{Communicational connectivity over bins}}]
\label{assum:communication_range}
The physical motion constraint of a robot is, in general, more stringent than its communicational constraint. 
From this, it can be assumed that if the transition of agents between bin $\mathcal{B}_i$ and $\mathcal{B}_l$ is allowed within a unit time instant, then the both bins are communicationally-connected, i.e., if $A_k[i,l] = 1$ then $C_k[i,l] = 1$. 
Note that we set $C_k[i,l] = 0$ if $A_k[i,l] = 0$.
This implies that the matrix $C_k$ is irreducible, as is $A_k$.  
The communication network over the agents is assumed to be strongly-connected \cite{Bandyopadhyay2017, Demir2015a}. 
Using distributed consensus algorithms \cite{Demir2015a, Bandyopadhyay2013a, Bandyopadhyay2014}, 
each agent can access necessary local information in its neighbour bins.
\end{assumption}

%
%

\begin{assumption}[{\emph{Pre-known Information}} \cite{Bandyopadhyay2017}]
\label{assum.desired_dist}
The desired swarm distribution $\Theta$, the motion constraint matrix $A_k$ (also $C_k$), and other pre-determined values such as variables regarding objective functions and user-design parameters (which will be introduced later) are known by all the agents before they begin a mission. 
\end{assumption}

\begin{assumption}[\emph{Agent's capability}\cite{Bandyopadhyay2017, Demir2015a}]\label{assum.nav} 
Each agent can determine the bin to which it belongs, and know the locations of neighbour bins so that it can navigate toward any of these bins. 
The agent is capable of collision avoidance behaviours against other agents or obstacles. 
\end{assumption}

\subsection{{Problem Statement}}\label{sec:design_goals}

The objective of the swarm distribution guidance problem considered in this paper is to distribute a set of agents $\mathcal{A}$ over a set of bins $\mathcal{B}$ by the Markov matrix $M^j_k$ in a manner that holds the following desired features:  

\begin{desired_feature} \label{goal.stability} The swarm distribution $\mu_k^{\star}$ asymptotically converges to the desired swarm distribution $\Theta$ as time instant $k$ goes to infinity.  
\end{desired_feature}

\begin{desired_feature} \label{goal.no_idle_switch} Transitions of the agents between the bins are controlled in a way that $M^j_k$ becomes close to $I$ as $\mu_k^{\star}$ converges to $\Theta$. This implies that the agents are settled down after $\Theta$ is achieved, and thus unnecessary transitions can be reduced. Moreover, the agents identify and compensate any partial loss or failure of the swarm distribution. 
\end{desired_feature}		
		
\begin{desired_feature} \label{goal.by_local_info} 
For each agent in bin $\mathcal{B}_i$, the information required for generating time-varying stochastic decision policies is not global information (e.g., $\mu_k^{\star}$) but locally available information within $\mathcal{A}_{\mathcal{N}_k(i)}$. 
Thereby, the resultant time-inhomogeneous Markov process is based on LICA, and has benefits such as a shorter timescale for obtaining new information (than GICA), the potential for an asynchronous process, etc.  
\end{desired_feature}

%

\begin{remark}
One of our main contributions is to provide Desired Feature \ref{goal.by_local_info} as well as to retain Desired Features \ref{goal.stability} and \ref{goal.no_idle_switch} by additionally adopting Assumption \ref{assum:communication_range}, which can be elicited from other assumptions in the existing literature. 
\end{remark}

\section{A Closed-loop-type Framework using Local Information}\label{sec:framework0}

This section proposes a LICA-based framework for the swarm distribution guidance problem. 
The framework is different from the recent closed-loop-type algorithms in \cite{Bandyopadhyay2017, Demir2015a} in the sense that 
they utilise the global information (e.g., the current swarm distribution in Equation (\ref{eqn:current_distribution})) for constructing a time-inhomogeneous Markov matrix, 
whereas ours uses the local information in Equation (\ref{eqn:local_current_distribution}). 
We present, in spite of using such relatively insufficient information, how the desired features described in the previous section can be achieved in the proposed framework.    
Before that, we introduce the biological idea, which is about decision-making mechanisms of a fish swarm, that inspires this framework to particularly attain Desired Feature \ref{goal.by_local_info}.  
In addition, we explicitly provide the design requirements for a Markov matrix in order for prospective users to easily incorporate their own specific objectives into this framework.

\subsection{The Biological Inspiration}

For a swarm of fishes, 
it has commonly been assumed that their crowdedness limits their perception ranges over other members, 
and their cardinality restricts the capacity for individual recognition \cite{Couzin2005}. 
How fishes end up with collective behaviours is different from the ways of other social species such as bees and ants, which are known to use recruitment signals for the guidance of the entire swarm \cite{Seeley1995, Keller2000}.  
Thus, in biology domain, a question naturally has arisen about the mechanism of fishes' decision-making in an environment where local information is only available and information transfer between members does not explicitly happen \cite{Partridge1982, Couzin2005, Becco2006, Couzin2002, Gautrais2008, Hoare2004}.


It has been experimentally shown that fishes' swimming activities vary depending on their perceivable neighbours.  
According to \cite{Partridge1982}, fishes have the tendency to maintain their statuses (e.g., position, speed, and heading angle) relative to those of other nearby fishes, which results in their organised formation structures.  
In addition, it is presented in \cite{Becco2006} that
spatial density of fishes has influences on both the minimum distances between them and the primary orientation of the fish school. 

Based on this knowledge, 
the works in \cite{Couzin2002, Couzin2005, Gautrais2008, Hoare2004} suggest individual-based models to further understand the collective behavioural mechanisms of fishes: for example, their repelling, attracting, and orientating behaviours \cite{Couzin2002, Gautrais2008}; 
how the density of informed fishes affects the elongation of the formation structure \cite{Couzin2005};
and group-size choices \cite{Hoare2004}. 
The common and fundamental characteristic of these models is that every agent maintains or adjusts its personal status with consideration of those of other individuals within its limited perception range.

As inspired by the understanding of fishes,
we believe that there must be an enhanced swarm distribution guidance approach  
in which each agent only needs to keep its relative status by using local information available from its nearby neighbours. 
In this approach, 
a global information is not necessary to be known by agents, 
and thereby the corresponding requirement of extensive information sharing over all the agents can be alleviated. 
\subsection{Fundamental Idea of the Proposed Approach}\label{sec:fundamental_idea}


\begin{figure}[tpb]
\centering
{\includegraphics[width=0.4\linewidth]{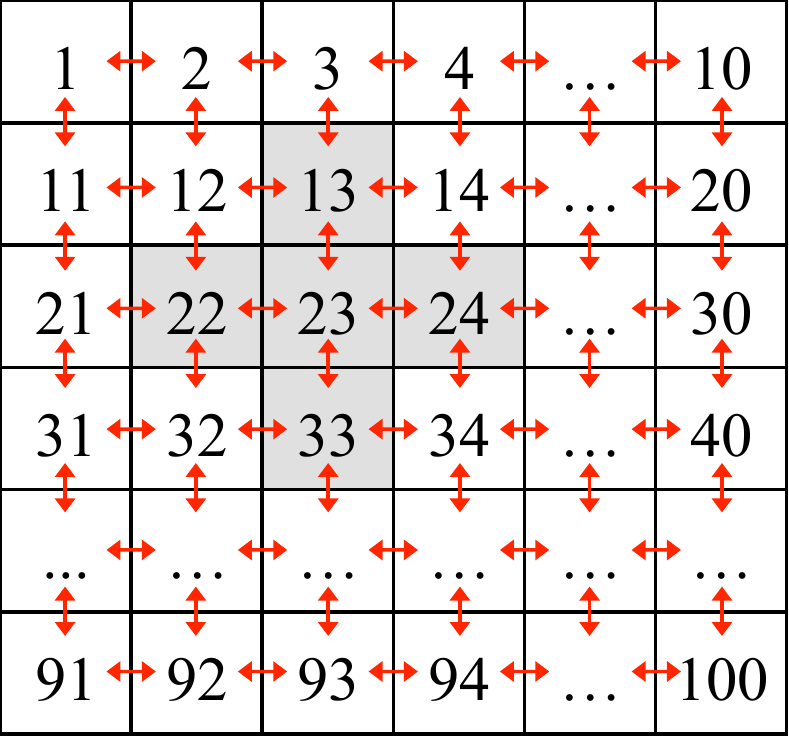}}
\caption{Examples of how to calculate $\bar{\mu}_k^{\star}[i]$: for bin $\mathcal{B}_{23}$, $\bar{\mu}_k^{\star}[23] = n_k[23] / ( n_k[13] + n_k[22] + n_k[23] + n_k[24] + n_k[33])$. 
In the proposed framework, agents in the bin only need to obtain the local information from other agents in its neighbour bins (shaded).
Note that each square indicates each bin, and the red arrow between two bins $\mathcal{B}_i$ and $\mathcal{B}_l$ means that $A_k[i,l] = 1$.  
}
\label{fig.comm_burden}
\end{figure}


Suppose that each agent in bin $\mathcal{B}_i$ is required to keep its local status $\bar{\mu}_k^{\star}[i]$, which we referred to as \emph{the current local swarm density at bin $\mathcal{B}_i$}, at the value of the corresponding \emph{locally-desired swarm density} $\bar{\Theta}[i]$. They are respectively defined as follows:
\begin{equation}
\label{eqn:local_dist}
		\bar{\mu}^{\star}_k[i] := \frac{n_k[i]}{\sum_{\forall l:\mathcal{B}_l \in \mathcal{N}_k(i)} n_k[l]}, 
\end{equation}
where $n_k[i]$ is the number of agents such that $a^j_k[i] = 1$; and
\begin{equation}
	\label{eqn:local_dist_desired}
	\bar{\Theta}[i] := \frac{\Theta [i]}{\sum_{\forall l:\mathcal{B}_l \in \mathcal{N}_k(i)} \Theta [l]}.
\end{equation}
We use the term $\bar{\mu}^{j}_k[i]$ as an estimate of $\bar{\mu}^{\star}_k[i]$ by agent $j$, which can be obtained through a distributed information consensus algorithm \cite{Demir2015a, Bandyopadhyay2013a, Bandyopadhyay2014}. 

The fundamental idea of the proposed approach is to make each agent $j$ in bin $\mathcal{B}_i$: 
\begin{enumerate}
\item[(i)] only need to estimate the difference of $\bar{\mu}^{\star}_k[i]$ and $\bar{\Theta}[i]$, which are both locally-available information within $\mathcal{N}_k(i)$; and
\item[(ii)] more reluctant to deviate from the current bin as the difference becomes smaller (i.e., $M^j_k[i,i] \rightarrow 1$ and $M^j_k[i,l] \rightarrow 0 \ \forall l $ as $\bar{\mu}^{j}_k[i] \rightarrow \bar{\Theta}[i]$). 
\end{enumerate}

Our proposed framework utilises the difference between $\bar{\mu}^{j}_k[i]$ and $\bar{\Theta}[i]$ as a local-information-based feedback gain, denoted by $\bar{\xi}^j_k[i]$, which is a scalar in $(0,1]$ that monotonically decreases as $\bar{\mu}^{j}_k[i]$ converges to $\bar{\Theta}[i]$.
For instance, this paper uses 
\begin{equation}\label{eqn:xi}
\bar{\xi}^j_k[i]:= \begin{cases}
		( \frac{|\bar{\Theta}[i] - \bar{\mu}^{j}_k[i]|}{\bar{\Theta}[i]} )^{\alpha} & \text{if $|\bar{\Theta}[i] - \bar{\mu}^{j}_k[i]| \le {\bar{\Theta}[i]}$}\\
		\epsilon_{\xi} & \text{if $( \frac{|\bar{\Theta}[i] - \bar{\mu}^{j}_k[i]|}{\bar{\Theta}[i]} )^{\alpha} < \epsilon_{\xi}$}\\
		1	& \text{otherwise}
		\end{cases}
\end{equation}
where $\alpha > 0$ and $\epsilon_{\xi} > 0$ are design parameters. 
We call this gain \emph{primary local-feedback gain} because it is utilised to control the primary guidance matrix $P_k^j$ (shown in the next subsection). 

\begin{remark}
Equation (\ref{eqn:local_dist}) is equivalent to the $i$-th element of the following vector:
\begin{equation}\label{eqn:local_current_distribution}
{\bar{\mu}_k(i) = }\frac{1}{|\mathcal{A}_{\mathcal{N}_k(i)}|} \sum_{\forall j \in \mathcal{A}_{\mathcal{N}_k(i)}} a^j_k.
\end{equation}
Namely, ${\bar{\mu}^{\star}_k[l]} = {\bar{\mu}_k(i)[l]}$ if $l = i$. 
Here, we intentionally introduce Equation (\ref{eqn:local_current_distribution}) for ease of comparison with the information required for feedback gains in the existing literature (e.g., Equation (\ref{eqn:current_distribution})). 
From this, it is implied that, 
in order for each agent in bin $\mathcal{B}_i$ to estimate ${\bar{\mu}_k(i)[i]}$ (i.e., the current local swarm density ${\bar{\mu}^{\star}_k[i]}$),  
the set of other agents whose information is necessary is restricted within $\mathcal{A}_{\mathcal{N}_k(i)}$.
That is, each agent needs to have neither a large perception radius nor an extensive information consensus process over the entire agents. 

\end{remark}


\subsection{A LICA-based Closed-loop-type Framework}\label{sec:framework}

This subsection presents our closed-loop-type framework based on locally-available information feedbacks. 
The basic form of the stochastic decision policy for agent $j$ in bin $\mathcal{B}_i$ is such that 
\begin{equation}\label{eqn:P}
        M^j_k[i,l] := \begin{cases}
                        		(1 - \omega^j_{k}[i]) P^j_k[i,l] + \omega^j_{k}[i]S^j_k[i,l] &\text{if $l = i$}\\
                        		(1 - \omega^j_{k}[i]) P^j_k[i,l] &\text{$\forall l \neq i$}.
                        \end{cases}
\end{equation}
Here, $\omega^j_k[i] \in [0,1)$ is the weighting factor to have different weights on the agent's {primary decision policy} $P^j_k[i,l] \in \mathbb{P}$ and {secondary decision policy} $S^j_k[i,l] \in \mathbb{P}$. It is defined as
 \begin{equation}\label{eqn:eta}
\omega^j_{k}[i] := {\exp(-\tau^j k)} \cdot {G^j_k[i]}  
\end{equation}
where $\tau^j$ is a design parameter; and $G^j_k[i] \in [0,1]$ is \emph{secondary local-feedback gain}, which is based on the difference between $\bar{\mu}^j_k[i]$ and $\bar{\Theta}[i]$.
Note that $\omega_k^j[i]$ is mainly affected by $G^j_k[i]$, while diminishing as time instant $k$ goes to infinity.

Equation (\ref{eqn:P}) can be represented in matrix form as
\begin{equation}\label{eqn:P_matrix}
M^j_k = (I - W^j_k) P^j_k + W^j_k S^j_k,
\end{equation}
where 
$P^j_k \in \mathbb{P}^{n_{bin} \times n_{bin}}$ and $S^j_k \in \mathbb{P}^{n_{bin} \times n_{bin}}$ are row-stochastic matrices, called \emph{primary guidance matrix} and \emph{secondary guidance matrix}, respectively.  
$W^j_k \in \mathbb{R}^{n_{bin} \times n_{bin}}$ is a diagonal matrix such that $\mathrm{diag}(W^j_k) = (\omega^j_k[1],...,\omega^j_k[n_{bin}])$. 
The stochastic state vector of agent $j$ is governed by the Markov process:
\begin{equation}\label{eqn:markov}
x^j_{k+1} = x^j_{k}M^j_k.
\end{equation}



For now, we claim that, in order for this Markov system to achieve Desired Features \ref{goal.stability}-\ref{goal.by_local_info}, $P^j_k$ must satisfy the following requirements.

\begin{requirement} $P^j_k$ is a matrix with row sums equal to one, i.e., 
\begin{equation}\tag{R1}\label{const:stochastic}
\sum_{l=1}^{n_{bin}} P^j_k[i,l] = 1, \forall i.
\end{equation}
In fact, $P^j_k$ needs to be row-stochastic, for which it should further hold that $P^j_k[i,l] \ge 0$, $\forall i,l$. 
Note that this constraint is implied by (\ref{const:irreducible}), which will be introduced later. 
\end{requirement}

\begin{requirement}
All diagonal elements are positive, i.e., 
\begin{equation}\tag{R2}\label{const:positive_diag}
P^j_k[i,i] > 0, \forall i.
\end{equation} 
\end{requirement}

\begin{requirement}
The stationary distribution of $P^j_k$ is the desired swarm distribution $\Theta$, i.e., 
	\begin{equation}\tag{R3}\label{const:reversible}
		\sum_{i=1}^{n_{bin}} \Theta[i] P^j_k[i,l] = \Theta[l], \forall l. 
	\end{equation}	
	With consideration of (\ref{const:stochastic}), 
	this requirement can be fulfilled by $\Theta[i] P^j_k[i,l] = \Theta[l] P^j_k[l,i]$, $\forall i$.
	A Markov process satisfying this property is said to be \emph{reversible}.  
\end{requirement}

\begin{requirement} $P^j_k$ is irreducible such that
\begin{equation}\tag{R4}\label{const:irreducible}
\begin{aligned}
P^j_{k}[i,l] > 0 & \quad \text{  if $C_{k}[i,l] = 1$}. \\
P^j_{k}[i,l] = 0 & \quad \text{ otherwise}.
\end{aligned}
\end{equation}
Note that $C_k$ is already assumed to be irreducible in Assumption \ref{assum:communication_range}. 
\end{requirement}


\begin{requirement}
{$P^j_k$ becomes close to $I$ as $\bar{\mu}^j_k$ converges to $\bar{\Theta}$}, i.e., 
\begin{equation}\tag{R5}\label{const:identity}
P^j_k[i,i] \to 1 \ \text{as $\bar{\mu}^j_k[i] \to \bar{\Theta}[i]$ (or $\bar{\xi}^j_k[i] \to 0$), $\forall i$}.
\end{equation}
\end{requirement}


Depending on the objectives of a user, $P^j_k$, $S^j_k$, $\bar{\xi}^j_k[i]$ and $G^j_k[i]$ can be designed differently under given specific constraints. 
As long as $P^j_k$ holds (\ref{const:stochastic})-(\ref{const:identity}) for all time instant $k$ and all agent $j \in \mathcal{A}$, the aforementioned desired features are achieved. 
Note that 
(\ref{const:stochastic})-(\ref{const:irreducible}) are associated with Desired Feature \ref{goal.stability}, whereas
(\ref{const:identity}) is with Desired Feature \ref{goal.no_idle_switch}.
The detailed analysis will be described in the next subsection. 

Every agent executes the following algorithm at every time instant.  
The detail regarding Line \ref{line:M_k}-\ref{line:G_k} will be presented in Section \ref{sec:implementation_example}, which shows examples of how to implement this framework. 

\begin{algorithm}
\caption{Decision making of agent $j$ at time instant $k$}\label{algorithm}
\begin{algorithmic}[1]

	\Statex \emph{// Obtain the local information}
	\State Identify the current bin $\mathcal{B}_i$;
	\State Identify neighbour bins $\mathcal{N}_k(i)$ (and $C_k[i,l] \ \forall l$);
	\State Compute $\bar{\Theta}[i]$ using (\ref{eqn:local_dist_desired});	
	\State Obtain $\bar{\mu}_k^j[i]$;
	\Statex \emph{// Generate the stochastic decision policy}	
	\State Compute $\bar{\xi}_k^j[i]$ (using (\ref{eqn:xi})); \label{line:xi}		
	\State Compute $P^j_k[i,l]$ $\forall l$; \label{line:M_k}
	\State Compute $S^j_k[i,l]$ $\forall l$; \label{line:Q_k}
	\State Compute $G^j_k[i]$; \label{line:G_k} 
	\State Compute $\omega^j_k[i]$ using (\ref{eqn:eta});
	\State Compute $M^j_k[i,l]$ $\forall l$ using (\ref{eqn:P});
	\Statex \emph{// Individually behave based on the policy}		
	\State Generate a random number $z \in \mathrm{unif}[0,1]$; \label{line:SSA_start}
	\State Select bin $\mathcal{B}_q$ such that 
	\Statex \quad \quad \quad ${\sum_{l=1}^{q-1} M_k^j[i,l] \le z < \sum_{l=1}^{q} M^j_k[i,l]}$;		
	\State Move to the selected bin;\label{line:SSA_end}
\end{algorithmic}
\label{alg:task_selection}
\end{algorithm}


\subsection{Analysis}\label{sec:analysis}

We first show that the Markov process in Equation (\ref{eqn:markov}) holds Desired Feature \ref{goal.stability} under the assumption that $P^j_k$ satisfies the requirements (\ref{const:stochastic})-(\ref{const:irreducible}) for each time instant. 
The stochastic state of agent $j$ at time instant $k \ge k_0$, 
governed by the Markov process from an arbitrary initial state $x^j_{k_0}$, can be written as:
\begin{equation}\label{eqn:U}
x^j_{k} = {x}^j_{k_0} U^j_{k_0,k} := {x}^j_{k_0} M^j_{k_0} M^j_{k_0+1} \cdots M^j_{k-1}. 
\end{equation}
For ease of analysis, we assume that every agent $j$ knows any necessary information correctly, i.e., $\bar{\mu}^{j}_k[i] = \bar{\mu}^{\star}_k[i]$.

\begin{theorem}\label{thm:converge}
Provided that the requirements (\ref{const:stochastic})-(\ref{const:irreducible}) are satisfied for all time instants $k \ge k_0$, 
it holds that $\lim_{k \to \infty} {x}^j_k = \Theta$ pointwise for all agents, irrespective of the initial condition.
\end{theorem}

\begin{proof}
This claim can be proved by following similar steps in proving \cite[Theorem 4]{Bandyopadhyay2017}.
The claim is true if $\lim_{k \to \infty} x^j_k = {x}^j_{k_0} \cdot \lim_{k \to \infty} U^j_{k_0,k} =  {x}^j_{k_0} \cdot \textbf{1}^{\top} \Theta = \Theta$. 
In order for that, the matrix product $U^j_{k_0,k}$ should (i) be {strongly ergodic} and (ii) have $\Theta$ as its unique limit vector, i.e., $\lim_{k \to \infty} U^j_{k_0,k} = \textbf{1}^{\top}\Theta$.
We will show that the two conditions are valid under the assumption that (\ref{const:stochastic})-(\ref{const:irreducible}) are satisfied.

Lemma \ref{lemma:P_U} in Appendix describes the characteristics of $M^j_k$ and $U^j_{k_0,k}$, which will be used for the rest of this proof. 
From this lemma, (a) $U^j_{k_0,k}$ is {primitive} (thus, regular); 
(b) there exists a positive lower bound $\gamma$ for $M_k^j$, $\forall k$; and
(c) $M^j_k$ is {asymptotically homogeneous}.  
Then, from \cite[Theorem 4.15, p.150]{Seneta1981} it follows that $U^j_{k_0,k}$ is {strongly ergodic}, which fulfils the condition (i).

Let $\textbf{e}_k \in \mathbb{P}^{n_{bin}}$ be the unique stationary distribution vector corresponding to $M^j_k$ (i.e., $\textbf{e}_k M^j_k = \textbf{e}_k$). 
Due to the prior condition (b) and the fact that (d) $M^j_k$ is {irreducible} for $\forall k \ge k_0$, 
it follows from \cite[Theorem 4.12, p.149]{Seneta1981} that 
the {asymptotical homogeneity} of $M^j_k$ with respect to $\Theta$ (i.e., $\lim_{k \to \infty} \Theta M^j_k = \Theta$) is equivalent to 
$\lim_{k \to \infty} \textbf{e}_k = \textbf{e}$ and $\Theta = \textbf{e}$, where $\textbf{e}$ is a limit vector. 
According to \cite[Corollary, p.150]{Seneta1981}, under the prior conditions (b) and (d), 
if $U^j_{k_0,k}$ is {strongly ergodic} with its unique limit vector $\textbf{v}$, then $\textbf{v} = \textbf{e}$. 
Hence, it turns out that the unique limit vector of $U^j_{k_0,k}$ is $\Theta$ (i.e, $\lim_{k \to \infty} U^j_{k_0,k} = \textbf{1}^{\top}\Theta$). 
Thereby, the condition (ii) is also fulfilled.
\end{proof}

Theorem \ref{thm:converge} implies that the stochastic state of any agent eventually converges to the desired swarm distribution, regardless of $S^j_k$, $G^j_k[i]$ and (\ref{const:identity}). 
In other words, even if (\ref{const:identity}) is not satisfied, the Markov system can converge to $\Theta$.
However, the system induces unnecessary transitions of agents even after being close enough to the desired swarm distribution, which means that Desired Feature \ref{goal.no_idle_switch} does not hold.


For now, we present that Desired Feature \ref{goal.no_idle_switch} can be obtained by (\ref{const:identity}) and Theorem \ref{theorem.converge_pursuing_local}, which will be described later. 
Suppose that, for every bin $\mathcal{B}_i$, $\bar{\mu}^{\star}_k[i]$ converges to and eventually reaches $\bar{\Theta}[i]$ at some time instant $k$. 
The following shows that at this moment it also holds that ${\mu}^{\star}_k$ reaches ${\Theta}$. 
From Equations (\ref{eqn:local_dist})-(\ref{eqn:local_dist_desired}) and the supposition of $\bar{\mu}^{\star}_k[i] = \bar{\Theta}[i]$ $\forall i$, it follows that  
$1/\bar{\Theta}[i] \cdot n_k[i] = \sum_{\forall j:\mathcal{B}_j \in \mathcal{N}_k(i)} n_k[j]$ $\forall i$. 
This can be rearranged as:
\begin{equation}
	\label{Def_B_matrix}
	\textbf{n}_k \cdot B  := \textbf{n}_k  \cdot (C_k - X)  = \textbf{0}
\end{equation} 
where $X \in \mathbb{R}^{n_{bin} \times n_{bin}}$ is a diagonal matrix such that $\text{diag}(X) =  ({1}/{\bar{\Theta}[1]}, {1}/{\bar{\Theta}[2]}, ..., {1}/{\bar{\Theta}[n_{bin}]})$; 
$C_k$ is the communicational connectivity matrix (in Definition \ref{def:comm_connect}); and $\textbf{n}_k \in \mathbb{R}^{n_{bin}}$ is a row vector such that the $i$-th element indicates $n_k[i]$, i.e., the number of agents in bin $\mathcal{B}_i$ at time instant $k$.

\begin{figure}[tpb]
\centering
\subfloat[Tree-type]{\includegraphics[width=0.4\linewidth]{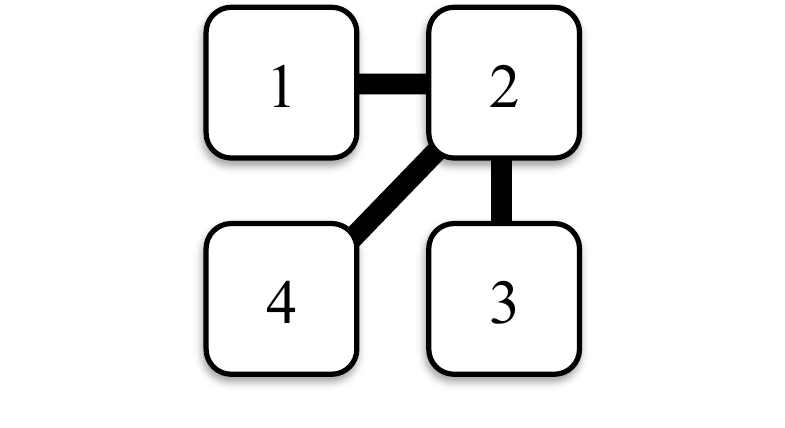} 
\label{fig.tree_type}}
\hfil
\subfloat[Strongly-connected]{\includegraphics[width=0.4\linewidth]{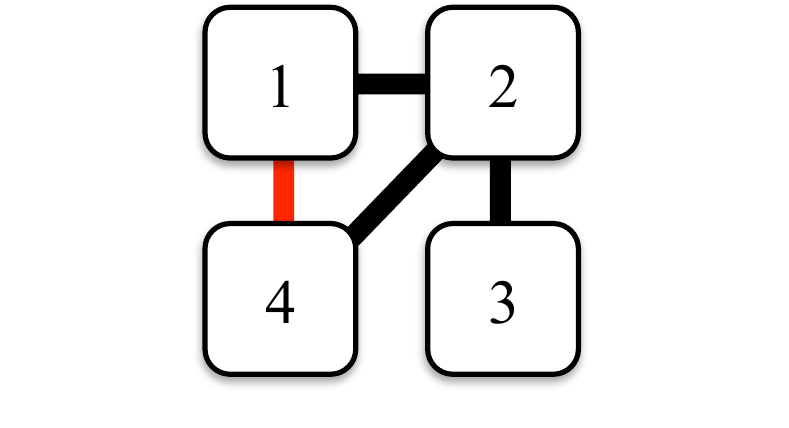}
\label{fig.tree_loop_type}}
\caption{Examples of simple bin topologies to help Lemma \ref{lemma.tree_type_rank} \& \ref{add_connect}: (a) tree-type; (b) strongly-connected. The red line in (b) indicates a newly-added route between bin $\mathcal{B}_1$ and $\mathcal{B}_4$ based on the topology in (a).}\label{fig.topology_example}
\end{figure}

\begin{lemma}
\label{lemma.tree_type_rank}
Given $n_{bin}$ bins communicationally-connected as a tree-type topology, the rank of its corresponding matrix $B$ in Equation (\ref{Def_B_matrix}) is $n_{bin}-1$. 
\end{lemma}

\begin{proof}
	The matrix $B \in \mathbb{R}^{n_{bin} \times n_{bin}}$ can be linearly decomposed into $n_e$ of the same-sized matrices $B_{(i,j)}$, 
	where $n_e$ is the number of edges in the underlying graph of $C_k$. 
	Here, $B_{(i,j)} \in \mathbb{R}^{n_{bin} \times n_{bin}}$ is a matrix such that 
	$B_{(i,j)}[i,i] = -\Theta[j]/\Theta[i]$ and $B_{(i,j)}[j,j] = -\Theta[i]/\Theta[j]$; 
	$B_{(i,j)}[i,j] = B_{(i,j)}[j,i] = 1$; and all the other entries are zero. 	
	For example, consider that four bins are given and connected as shown in Figure \ref{fig.topology_example}(a). 
	Clearly, $B = B_{(1,2)} + B_{(2,3)} + B_{(2,4)}$, where 	
	\begin{equation*}
	B = 	 
	\begin{pmatrix} 
	-\frac{\Theta[2]}{\Theta[1]} & 1 & 0 & 0 \\
	1 & -\frac{\Theta[1] + \Theta[3] + \Theta[4]}{\Theta[2]} & 1 & 1 \\
	0 & 1 & -\frac{\Theta[2]}{\Theta[3]} & 0 \\
	0 & 1 & 0 & -\frac{\Theta[2]}{\Theta[4]} 		
	\end{pmatrix},
	\end{equation*}	

	\begin{equation*}
	B_{(1,2)} = 	 
	\begin{pmatrix} 
	-\frac{\Theta[2]}{\Theta[1]} & 1 & 0 & 0 \\
	1 & -\frac{\Theta[1]}{\Theta[2]} & 0 & 0 \\
	0 & 0 & 0 & 0 \\
	0 & 0 & 0 & 0		
	\end{pmatrix},
	\end{equation*}	

	\begin{equation*}
	B_{(2,3)} = 	 
	\begin{pmatrix} 
	0 & 0 & 0 & 0 \\
	0 & -\frac{\Theta[3]}{\Theta[2]} & 1 & 0 \\
	0 & 1 & -\frac{\Theta[2]}{\Theta[3]} & 0 \\
	0 & 0 & 0 & 0 	
	\end{pmatrix},
	\end{equation*}	

	\begin{equation*}
	B_{(2,4)} = 	 
	\begin{pmatrix} 
	0 & 0 & 0 & 0 \\
	0 & -\frac{\Theta[4]}{\Theta[2]} & 0 & 1 \\
	0 & 0 & 0 & 0 \\
	0 & 1 & 0 & -\frac{\Theta[2]}{\Theta[4]} 		
	\end{pmatrix}.
	\end{equation*}

	It is trivial that the rank of every $B_{(i,j)}$ is one, and the matrix has only one linearly independent column vector, denoted by $v_{(i,j)}$. 	
	Without loss of generality, we consider $v_{(i,j)} \in \mathbb{R}^{n_{bin}}$ as a column vector such that the $i$-th entry is $-\frac{1}{\Theta[i]}$, the $j$-th entry is $\frac{1}{\Theta[j]}$, and the others are zero: for an instance, $v_{(1,2)} = [-\frac{1}{\Theta[1]}, \frac{1}{\Theta[2]}, 0, 0]^{\top}$. 
	
	It is obvious that $v_{(i,j)}$ and $v_{(k,l)}$ are linearly independent when the bin pairs $\{i,j\}$ and $\{k,l\}$ are different. 
	This implies that the number of linearly independent column vectors of $B$ is the same as that of edges in the topology.
	Hence, for a tree-type topology of $n_{bin}$ bins, since there exist $n_{bin}-1$ edges, the rank of the corresponding matrix $B$ is $n_{bin}-1$. 
\end{proof}

\begin{lemma}
\label{add_connect}
Given a strongly-connected topology of bins, 
the rank of its corresponding matrix $B$ is not affected by adding a new edge that directly connects any two existing bins.  
\end{lemma}

\begin{proof}
We will show that this claim is valid even when a tree-type topology is given, as it is a sufficient condition for strong-connectivity. 
Given the tree-type topology in Figure \ref{fig.topology_example}(a), suppose that bin $\mathcal{B}_1$ and $\mathcal{B}_4$ are newly connected. 
Then, the new topology becomes as shown in Figure \ref{fig.topology_example}(b), and it has new corresponding matrix $B_{new}$, 
where $B_{new} = B + B_{(1,4)}$. 
As explained in the proof of Lemma \ref{lemma.tree_type_rank}, 
the rank of $B_{(1,4)}$ is one and it has only linearly independent vector $v_{(1,4)}$. 
However, this vector can be produced as a linear combination of the existing $v$ vectors of $B$ (i.e., $v_{(1,4)} = v_{(1,2)} + v_{(2,4)}$).
Thus, the rank of $B_{new}$ retains that of $B$. 
Without loss of generality, this implies that
the rank of $B$ of a given strongly-connected topology is not affected by adding a new edge that directly connects any two existing bins.  
\end{proof}

Thanks to Lemma \ref{lemma.tree_type_rank} and \ref{add_connect}, 
we end up with the following corollary and theorem: 
\begin{corollary}
\label{col.rank_m_1}
Given $n_{bin}$ bins that are communicationally strongly-connected,
the rank of its corresponding $B$ is $n_{bin}-1$. 
\end{corollary}




\begin{theorem}
\label{theorem.converge_pursuing_local}
Given $n_{bin}$ bins that are communicationally strongly-connected,
convergence of $\bar{\mu}^{\star}_k$ to $\bar{\Theta}$ is equivalent to convergence of $\mu_k^{\star}$ to $\Theta$. 
\end{theorem}

\begin{proof}
From Equation (\ref{eqn:local_dist_desired}), it can be said that $\Theta \cdot B = \textbf{0}$.
When $\bar{\mu}^{\star}_k[i]$ is assumed to converge to $\bar{\Theta}[i]$ at some time instant $k$ for every bin $\mathcal{B}_i$, 
Equation (\ref{Def_B_matrix}) is valid (i.e., $\textbf{n}_k \cdot B = 0$). 
Since the nullity of $B$ is one, due to Corollary \ref{col.rank_m_1}, there is only one linearly-independent row-vector $\textbf{a} \in \mathbb{R}^{n_{bin}}$ such that $\textbf{a} \cdot  B  = \textbf{0}$. 
Hence, it is obvious that $\textbf{n}_k = \epsilon \cdot \Theta$, where $\epsilon$ is an arbitrary scalar value.  
This implies that $\mu^{\star}_k[i] = n_k[i]/n^{\mathcal{A}}_k = \Theta[i], \ \forall i: \mathcal{B}_i \in \mathcal{B}$.
Therefore, convergence of $\mu_k^{\star}$ to $\Theta$ is equivalent to convergence of $\bar{\mu}^{\star}_k$ to $\bar{\Theta}$. 
\end{proof}

From this theorem and (\ref{const:identity}), Desired Feature \ref{goal.no_idle_switch} finally holds. 

\begin{corollary}\label{lemma:M_k}
If $P^j_k$ satisfies (\ref{const:identity}), it can be said from Theorem \ref{theorem.converge_pursuing_local} that $P^j_k$ becomes $I$ as $\mu_k^{\star}$ converges to $\Theta$. 
And this is also the case for the Markov process $M^j_k$, which satisfies Desired Feature \ref{goal.no_idle_switch}.
\end{corollary}

In order for each agent $j$ in bin $\mathcal{B}_i$ to generate the time-varying stochastic decision policy $M^j_k[i,l]$ in Equation (\ref{eqn:P}), 
the agent only needs to obtain its local information within $\mathcal{A}_{\mathcal{N}_k(i)}$. 
Therefore, Desired Feature \ref{goal.by_local_info} is also achieved.

%
%
%

\begin{remark}[\emph{Robustness against dynamic changes of agents and those of bins}]
The proposed framework is robust with against dynamic changes in the number of agents and bins. 
Similarly to what is claimed in \cite[Remark 8]{Bandyopadhyay2017}, as each agent behaves based on its current bin location and local information in a memoryless manner, Desired Features \ref{goal.stability}-\ref{goal.by_local_info} in the proposed framework won't be affected by inclusion or exclusion of agents in a swarm.     
Furthermore, provided changes on bins are perceived by at least nearby agents in the corresponding neighbour bins, robustness against those changes can be hold in the proposed framework.
This is because agents in bin $\mathcal{B}_i$ utilise only local information such as $\bar{\Theta}[i]$ and $\bar{\mu}_k^j[i]$, and are not required to know any information from other far-away bins. 
Moreover, the proposed framework does not need to recalculate $\Theta$, reflecting such changes on bins, so that $\sum_{\forall i} \Theta[i] = 1$ because computing $\bar{\Theta}[i]$ in (\ref{eqn:local_dist_desired}) includes normalisation of $\Theta$. 
\end{remark}

\section{Implementation Examples}\label{sec:implementation_example}
\subsection{Example I: Minimising Travelling Expenses}\label{sec:min_cost}
This section provides examples on implementations of the framework proposed. 
In particular, this subsection addresses a problem of minimising travelling expenses of agents during convergence to a desired swarm distribution. 
   
%

This problem can be defined as: given a cost matrix $E_k \in \mathbb{R}^{n_{bin} \times n_{bin}}$ in which each element ${E}_k[i,l]$ represents the travelling expense of an agent from bin $\mathcal{B}_i$ to $\mathcal{B}_l$, find $P^j_k$ such that
\begin{equation}\tag{P1}\label{problem_min_cost}
\min \sum_{i=1}^{n_{bin}} \sum_{l=1}^{n_{bin}} {E}_k[i,l]P^j_k[i,l]
\end{equation}
subject to (\ref{const:stochastic})-(\ref{const:identity}) and 
	\begin{equation}\label{lp:offdiagonals}
	\begin{split}		
		\epsilon_M {\Theta[l]} {f(\bar{\xi}^j_k[i],\bar{\xi}^j_k[l])} f(E_k[i,l]) \le P^j_k[i,l] \\
		 \text{if $C_k[i,l] = 1$, $\forall i \neq l$}		
		 \end{split}
	\end{equation}	
where $\epsilon_M \in (0,1]$ is a design parameter. 
$f(\bar{\xi}^j_k[i],\bar{\xi}^j_k[l]) \in (0,1]$ is set by 
\begin{equation}\label{eqn:ex_xi}
{f(\bar{\xi}^j_k[i],\bar{\xi}^j_k[l])} = \max(\bar{\xi}^j_k[i],\bar{\xi}^j_k[l])
\end{equation}
so that the value monotonically increases with regard to increase of either $\bar{\xi}^j_k[i]$ or $\bar{\xi}^j_k[l]$ and diminishes as $\bar{\xi}^j_k[i]$ and $\bar{\xi}^j_k[l]$ simultaneously reduces.
This value controls the lower bound of $P^j_k[i,l]$ in Equation (\ref{lp:offdiagonals}). 
$\Theta[l]$ enables agents in bin $\mathcal{B}_i$ to be distributed over its neighbour bins in proportion to the desired swarm distribution. 
$f(E_k[i,l]) \in (0,1]$ is a scalar that monotonically decreases as $E_k[i,l]$ increases (see Equation (\ref{eqn:f_E_k}) for instances),
encouraging agents in bin $\mathcal{B}_i$ to avoid spending higher transition expenses. 
Note that we assume that $E_k$ is symmetric; $E_k[i,l] > 0$ if $A_k[i,l] = 1$; and its diagonal entries are zero.

\begin{corollary}
The optimal matrix $P^j_k$ of the problem (\ref{problem_min_cost}) is given by:
$\forall i, l \in \{1,...,n_{bin}\}$ and $i \neq l$, 
\begin{equation}\label{eqn:optimal_offdiag}
P^j_k[i,l] = \begin{cases}	
	\epsilon_M {\Theta[l]} {f(\bar{\xi}^j_k[i],\bar{\xi}^j_k[l])} f(E_k[i,l])    \quad   \text{if $C_k[i,l] = 1$}\\
	0 \quad \quad \quad \quad \quad \quad \quad \quad \quad\quad\quad\quad\quad  \text{ otherwise}
\end{cases}
\end{equation}
and $\forall i = l$, 
\begin{equation}\label{eqn:optimal_diag}
P^j_k[i,i] = 1 - \sum_{\forall l \neq i} P^j_k[i,l].
\end{equation}
\end{corollary}

\begin{proof}
We can prove this by following the proof of \cite[Corollary 1]{Bandyopadhyay2017}.
Suppose that the problem is only subject to (\ref{const:irreducible}) and (\ref{lp:offdiagonals}), without (\ref{const:stochastic})-(\ref{const:reversible}) and (\ref{const:identity}). 
Then, the off-diagonal elements of an optimal matrix should be their corresponding lower bounds in (\ref{lp:offdiagonals}) if $C_k[i,l] = 1$.
The diagonal elements of the matrix do not affect the objective function due to the fact that $E_k[i,i] = 0, \forall i$. 
Accordingly, the matrix $P^j_k$ that holds (\ref{eqn:optimal_offdiag}) and (\ref{eqn:optimal_diag}) is also an optimal matrix for the simplified problem.

Let us now consider (\ref{const:stochastic})-(\ref{const:reversible}) and (\ref{const:identity}).
Since $\epsilon_M$, $f(\bar{\xi}^j_k[i],\bar{\xi}^j_k[l])$ and $f(E_k[i,l])$ are upper-bounded by $1$ and $\sum_{\forall l \neq i} \Theta[l] < 1$, $P^j_k[i,i]$ in  (\ref{eqn:optimal_diag}) is always positive for all $i$, which fulfils  (\ref{const:positive_diag}). 
It is also obvious that (\ref{const:stochastic}) is satisfied by Equation (\ref{eqn:optimal_diag}). 
From Equation (\ref{eqn:optimal_offdiag}), it holds that $\Theta[i] P^j_k[i,l] = \Theta[l] P^j_k[l,i]$, complying with (\ref{const:reversible}).
Since (\ref{const:stochastic})-(\ref{const:irreducible}) are satisfied, the Markov process is converging to a desired distribution due to Theorem \ref{thm:converge}. 
Noting that $f(\bar{\xi}^j_k[i],\bar{\xi}^j_k[l])$ diminishes as $\bar{\xi}_k^j$ gets close to zero (i.e., $\bar{\mu}_k^{j} \to \bar{\Theta}$), (\ref{const:identity}) is also fulfilled by  Equation (\ref{eqn:optimal_diag}).
Hence, $P^j_k$ is the optimal solution for the problem (\ref{problem_min_cost}). 
\end{proof}

For reducing unnecessary transitions of agents during this process, 
it is favourable that agents in bin $\mathcal{B}_i$ such that $\bar{\mu}^j_k[i] \le \bar{\Theta}[i]$ (i.e., underpopulated) do not deviate.
To this end, we set 
$S^j_k = I$ and $G^j_k[i]$ as follows \cite{Bandyopadhyay2017}:
\begin{equation}\label{eqn:G_k_P1}
G^j_k[i] := \frac{\exp(\beta (\bar{\Theta}[i] - \bar{\mu}^j_k[i]) )}{\exp(\beta |\bar{\Theta}[i]-\bar{\mu}^j_k[i]| )}.
\end{equation}
The gain value is depicted in Figure \ref{fig.G_P1P3}(a) with regard to $\beta$. 

\begin{remark}[\emph{Increase of Convergence Rate}]\label{remark3}
Due to the fact that $\sum_{\forall l \neq i} P^j_k[i,l] \le \sum_{\forall l:B_l \in \mathcal{N}_k(i) \setminus \mathcal{B}_i} \Theta[l]$ from Equation (\ref{eqn:optimal_offdiag}), 
the total outflux of agents from bin $\mathcal{B}_i$ becomes smaller as the bin has fewer connections with other bins. 
This eventually makes the convergence rate of the Markov process slower. 

Adding an additional variable into $P^j_k[i,l]$ in (\ref{eqn:optimal_offdiag}) does not affect the obtainment of Desired Features \ref{goal.stability}-\ref{goal.by_local_info} as long as $P^j_k$ satisfies (\ref{const:stochastic})-(\ref{const:identity}). 
Thus, in order to enhance the convergence rate under the requirements, 
one can add 
\begin{equation}
\epsilon_{{\Theta}} := \min \{ \frac{1}{\sum_{\forall s: \mathcal{B}_s \in \mathcal{N}_k(i) \setminus \mathcal{B}_i}  \Theta[s]}, \frac{1}{\sum_{\forall s: \mathcal{B}_s \in \mathcal{N}_k(l)  \setminus \mathcal{B}_l }  \Theta[s]}  \}
\end{equation}
into $P^j_k[i,l]$, as follows: 
\begin{equation}\label{eqn:optimal_offdiag2}
P^j_k[i,l] = \begin{cases}	
	\epsilon_{{\Theta}}  \epsilon_M {\Theta[l]} {f(\bar{\xi}^j_k[i],\bar{\xi}^j_k[l])} f(E_k[i,l])\\
	    \quad \quad \quad \quad \quad \quad \quad \quad \quad\quad\quad\quad\quad    \text{if $C_k[i,l] = 1$}\\
	0 \quad \quad \quad \quad \quad \quad \quad \quad \quad\quad\quad\quad  \text{ otherwise,}
\end{cases}
\end{equation}
which can be substituted for Equation (\ref{eqn:optimal_offdiag}). 
\end{remark}

\begin{algorithm}
\caption{Subroutine of Algorithm \ref{algorithm} (Line \ref{line:M_k}--\ref{line:G_k}) for P1}\label{algorithm_P1}
\begin{algorithmic}[1]

	\State Compute $P^j_k[i,l]$ $\forall l$ using (\ref{eqn:optimal_offdiag}) (or (\ref{eqn:optimal_offdiag2})) and (\ref{eqn:optimal_diag});
	\State Set $S^j_k[i,i] = 1$ and $S^j_k[i,l] = 0, \forall l \neq i$
	\State Compute $G^j_k[i]$ using (\ref{eqn:G_k_P1});	

\end{algorithmic}
\label{alg:task_selection}
\end{algorithm}

\subsection{Example II: Maximising Convergence Rate within Upper Flux Bounds}\label{sec:max_conv}


This subsection presents an example in which the specific objective is to maximise the convergence rate under upper bounds regarding transitions of agents between bins, denoted by \emph{upper flux bounds}. 
The bounds can be interpreted as safety constraints in terms of collision avoidance and congestion: higher congestions may induce higher collisions amongst agents, which may bring unfavourable effects on system performance. 
A similar problem is addressed by an open-loop-type algorithm in \cite{Berman2009}, where transitions of agents are limited only at a desired swarm distribution.  
This restriction is not for considering the aforementioned safety constraints, but rather for mitigating the trade-off between convergence rate and long-term system efficiency. 

For the sake of imposing upper flux bounds during the entire process, 
we consider the following one-way flux constraint:
\begin{equation}\label{const:flux_upper_bound}
n_k[i] P^j_k[i,l] \le c_{(i,l)}, \quad \forall i, \forall l \neq i. 
\end{equation}
This means that the number of agents moving from bin $\mathcal{B}_i$ to $\mathcal{B}_l$ is upper-bounded by $c_{(i,l)}$. 
The bound value is assumed to be very small with consideration of mission environments such as the number of agents, the number of bins, and their topology. 
Otherwise, all the agents can be distributed over the bins very soon so that the upper flux bounds become meaningless, and thus the corresponding problem can be trivial. 

Regarding the convergence rate of a Markov chain, there are respective analytical methods depending on whether it is time-homogeneous or time-inhomogeneous. 
For a time-homogeneous Markov chain, if the matrix is irreducible, 
the second largest eigenvalue of the matrix 
is used as an index indicating its asymptotic convergence rate \cite[p.389]{Sebbane2011}.
In contrast, for a time-inhomogeneous Markov chain, \emph{coefficients of ergodicity} can be utilised as a substitute for the second largest eigenvalue, which is not useful for this case \cite{Ipsen2011}. 
Particularly, this paper uses the following \emph{proper} coefficient of ergodicity, amongst others:
\begin{definition}(\emph{Coefficient of Ergodicity} \cite[pp. 136--137]{Seneta1981}).
Given a stochastic matrix $\mathcal{M} \in \mathbb{P}^{n \times n}$, a (proper) coefficient of ergodicity $0 \le \tau(\mathcal{M}) \le 1$ can be defined as:
\begin{equation}\label{eqn:coeffi_ergodic}
\tau(\mathcal{M}) := \max_{\forall s}\max_{\forall i, \forall l} |\mathcal{M}[i,s] - \mathcal{M}[l,s]|.
\end{equation}
A coefficient of ergodicity is said to be \emph{proper} if $\tau(\mathcal{M}) = 0$ if and only if $\mathcal{M} = \textbf{1}^{\top}\cdot v$, where $v \in \mathbb{P}^n$ is a row-stochastic vector. 
\end{definition} 

The convergence rate of a time-inhomogeneous Markov chain $\mathcal{M}_k \in \mathbb{P}^{n \times n}$, $\forall k > 1$ can be maximised by minimising $\tau(\mathcal{M}_k)$ at each time instant $k$, thanks to \cite[Theorem 4.8, p.137]{Seneta1981}: 
$\tau(\mathcal{M}_1\mathcal{M}_2 \cdots \mathcal{M}_r) \le \prod_{k=1}^{r} \tau(\mathcal{M}_k)$. 
Hence, the objective of the specific problem considered in this subsection can be defined as: find $P^j_k$ such that
\begin{equation} \label{eqn:problem2_original}
\mathrm{min} \ \tau(P^j_k)
\end{equation}
subject to (\ref{const:stochastic})-(\ref{const:identity}) and (\ref{const:flux_upper_bound}).

\begin{remark}[\emph{Advantages of the coefficient of ergodicity in (\ref{eqn:coeffi_ergodic})}]
Other proper coefficients in \cite[p. 137]{Seneta1981} such as 
\begin{equation*}
\tau_1(\mathcal{M}) = 1 - \min_{i,l} \sum_{\forall s} \min \left( \mathcal{M}[i,s], \mathcal{M}[l,s] \right)
\end{equation*} or
\begin{equation*}
\tau_2(\mathcal{M}) = 1 - \sum_{\forall s} \min_{\forall i} \left( \mathcal{M}[i,s] \right).
\end{equation*} 
may have the trivial case such that $\tau_1(P^j_k) = 1$ (or $\tau_2(P^j_k) = 1$) for some time instant $k$, when they are applied to this problem.  
This is because, given a strongly-connected topology $C_k$, there may exist a pair of bins $\mathcal{B}_i$ and $\mathcal{B}_l$ such that $P^j_k[i,s] = 0$ or $P^j_k[l,s]=0$, $\forall s$. 
To avoid this trivial case, the work in \cite{Bandyopadhyay2017} instead utilises $\tau_1 ( (P^j_k)^{d_{C_k}} )$ as the proper coefficient of ergodicity, 
where $d_{C_k}$ denotes the diameter of the underlying graph of $C_k$. 
However, this implies that agents in bin $\mathcal{B}_i$ are required to additionally access the information from other bins beside $\mathcal{N}_k(i)$, causing additional communicational costs. The coefficient of ergodicity in (\ref{eqn:coeffi_ergodic}) does not suffer this issue. 
Note that $\tau(\mathcal{M}) \le \tau_1(\mathcal{M}) \le \tau_2(\mathcal{M})$\cite[p. 137]{Seneta1981}. 
\end{remark}

Finding the optimal solution for the problem (\ref{eqn:problem2_original}) is another challenging issue, which can be called \emph{fastest mixing Markov chain problem}. 
Since the purpose of this section is to show an example of how to implement our proposed framework, 
we heuristically address this problem at this moment. 
 
Suppose that matrix $P^j_k$ satisfying (\ref{const:stochastic})-(\ref{const:identity}) is given, and the topology of bins is not fully-connected.
Since the matrix is non-negative and there exists at least one zero-value entry in each column, 
the coefficient of ergodicity can be said as $\tau(P^j_k) = \max_{\forall i,\forall s} (P^j_k[i,s])$. 
Assuming that $\max_{\forall l \neq i} P^j_k[i,l] \le 1/|\mathcal{N}_k(i)|$, which is generally true due to the smallness of $c_{(i,l)}$, 
it turns out that each diagonal element of $P^j_k$ is the largest value in each row.
Thus, we can say that $\tau(P^j_k) = \max_{\forall i} P^j_k[i,i]$. 
The objective function of this problem can be said as $\mathrm{maxmin}_{\forall i} \sum_{\forall l \neq i} P^j_k[i,l]$
because minimising the maximum diagonal element of a stochastic matrix is equivalent to maximising the minimum row-sum of its off-diagonal elements.

We turn now to the constraints (\ref{const:stochastic})-(\ref{const:identity}) and (\ref{const:flux_upper_bound}). 
In order to comply with (\ref{const:reversible}), we initially set $P^j_k[i,l] = \Theta[l]Q^j_k[i,l]$, where $Q^j_k$ is a symmetric matrix that we will design now. 
The constraint (\ref{const:flux_upper_bound}), (\ref{const:irreducible}), and the symmetricity of $Q_k$ are integrated into the following constraint: $\forall i$, $\forall l \neq i$, 
\begin{equation}\label{const:irreducible_2}
\begin{aligned}
 \min(\frac{c_{(i,l)}}{n_k[i]\Theta[l]},\frac{c_{(l,i)}}{n_k[l]\Theta[i]}) \ge Q^j_{k}[i,l] > 0 & \quad \text{  if $C_{k}[i,l] = 1$} \\
Q^j_{k}[i,l] = 0 & \quad \text{ otherwise}.
\end{aligned}
\end{equation}
For (\ref{const:positive_diag}) and (\ref{const:identity}), we set the diagonal entries of $P^j_k$ as
\begin{equation*}
	P^j_k[i,i] \ge 1 - \bar{\xi}^j_k[i], \quad \forall i. 
\end{equation*}
This can be rewritten, with consideration of (\ref{const:stochastic}) (i.e., $\sum_{l=1}^{n_{bin}} \Theta[l]Q^j_k[i,l] = 1, \forall i$), as
\begin{equation}\label{const:new2}
	\sum_{\forall l \neq i} \Theta[l] Q^j_k[i,l] \le \bar{\xi}^j_k[i], \quad \forall i. 
\end{equation}
Then, the reduced problem can be defined as: find $Q^j_k$ such that 
\begin{equation}\tag{P2}\label{problem_max_conv}
\mathrm{maxmin}_{\forall i} \sum_{\forall l \neq i} \Theta[l]Q^j_k[i,l]
\end{equation}
subject to (\ref{const:irreducible_2}) and (\ref{const:new2}).


The algorithm for this problem is shown in Algorithm \ref{algorithm_P2}.
If we neglect (\ref{const:new2}), an optimal solution can be obtained by making $Q^j_k[i,l]$ equal to its upper bound of (\ref{const:irreducible_2}) (Line \ref{P2:line3}).
However, this solution may not hold (\ref{const:new2}). 
Thus, we lower the entries of $Q^j_k$ to satisfy (\ref{const:new2}), while keeping them symmetric and as higher as possible (Line \ref{P2:line9}--\ref{P2:line14}). 
In details, Line \ref{P2:line9} (or Line \ref{P2:line12}) ensures the constraint (\ref{const:new2}) for each bin $\mathcal{B}_i$ in a way that, if this is not the case, obtains the necessary lowering factor $\bar{\epsilon}_Q'[i]$ (or $\bar{\epsilon}_Q[i]$). 
In order to keep $Q^j_k$ as higher as possible, we temporarily take ${\epsilon}_Q'[i,l]$ as the maximum value of $\{\bar{\epsilon}_Q'[i],\bar{\epsilon}_Q'[l]\}$ (Line \ref{line_epsilon_1}).
After curtailing $Q^j_k[i,l]$ by applying ${\epsilon}_Q'[i,l]$, we obtain the corresponding lowering factor again (Line \ref{P2:line11}--\ref{P2:line12}). 
The minimum value is taken for both maintaining $Q^j_k$ symmetric and satisfying (\ref{const:new2}) (Line \ref{line_epsilon}). 
Then, the corresponding stochastic decision policy is generated based on the resultant $Q^j_k$ (Line \ref{P2:line14}--\ref{P2:end}). 
Note that we set $G^j_k[i] = 0$ for all time instants, all bins, and all agents, so $M^j_k = P^j_k$.

\begin{algorithm}
\caption{Subroutine of Algorithm \ref{algorithm} (Line \ref{line:M_k}) for P1}\label{algorithm_P2}
\begin{algorithmic}[1]

	\Statex \emph{// Initialise $P^j_k$}
	\State $P^j_k[i,l] = 0$, $\forall l \in \{1,2,...,n_{bin}\}$;
	\Statex \emph{// Compute $Q^j_k$ satisfying (\ref{const:irreducible_2})}		
		\State $Q^j_k[i,l] = \min(\frac{c_{(i,l)}}{n_k[i]\Theta[l]},\frac{c_{(l,i)}}{n_k[l]\Theta[i]})$, $\forall \mathcal{B}_l \in \mathcal{N}_k(i) \setminus \{\mathcal{B}_i\}$; \label{P2:line3}	
	\Statex \emph{// Lower $Q^j_k$ to satisfy	 (\ref{const:new2})}		
	\State $\bar{\epsilon}_Q'[i] = \min (\frac{\bar{\xi}^j_k[i]}{\sum_{\forall l \neq i} \Theta[l] Q^j_k[i,l]}, 1)$; \label{P2:line9}	
	\State $\epsilon_Q'[i,l] = \max (\bar{\epsilon}_Q'[i], \bar{\epsilon}_Q'[l])$, $\forall \mathcal{B}_l \in \mathcal{N}_k(i) \setminus \{\mathcal{B}_i\}$; \label{line_epsilon_1}
	\State $Q^j_k[i,l] := \epsilon_Q'[i,l] Q^j_k[i,l]$, $\forall \mathcal{B}_l \in \mathcal{N}_k(i) \setminus \{\mathcal{B}_i\}$;\label{P2:line11}
	\State $\bar{\epsilon}_Q[i] = \min (\frac{\bar{\xi}^j_k[i]}{\sum_{\forall l \neq i} \Theta[l] Q^j_k[i,l]}, 1)$;\label{P2:line12}
	\State $\epsilon_Q[i,l] = \min (\bar{\epsilon}_Q[i], \bar{\epsilon}_Q[l])$, $\forall \mathcal{B}_l \in \mathcal{N}_k(i) \setminus \{\mathcal{B}_i\}$; \label{line_epsilon}
	\State $Q^j_k[i,l] := \epsilon_Q[i,l] Q^j_k[i,l]$, $\forall \mathcal{B}_l \in \mathcal{N}_k(i) \setminus \{\mathcal{B}_i\}$;
	\Statex \emph{// Compute $P^j_k$}
	\State $P^j_k[i,l] = \Theta[l]Q^j_k[i,l]$, $\forall \mathcal{B}_l \in \mathcal{N}_k(i) \setminus \{\mathcal{B}_i\}$;\label{P2:line14}
	\State $P^j_k[i,i] = 1 - \sum_{\forall l \neq i} P^j_k[i,l]$; \label{P2:end}
\end{algorithmic}
\end{algorithm}
\subsection{Example III: Local-information-based Quorum Model}\label{sec:quorum}
This subsection shows that the proposed framework is able to incorporate a quorum model, which is introduced in \cite{Halasz2007, Hsieh2008}. 
In this model, if a bin is overpopulated above a certain level of predefined threshold called \emph{quorum}, 
the probabilities that agents in the bin move to neighbour bins are temporarily increased, rather than following given $P^j_k$.   
This feature eventually brings an advantage to the convergence performance of the swarm.

To this end, we set the secondary guidance matrix $S^j_k$ as follows: $\forall i, l \in \{1,...,n_{bin}\}$ and $\forall j \in \mathcal{A}$, 
\begin{equation}\label{eqn:quorum_Q}
S^j_k[i,l] := \begin{cases}	
	{1}/{|\mathcal{N}_k(i)|}   \quad\quad   \text{if $C_k[i,l] = 1$}\\
	0 \quad \quad \quad \quad \quad \      \text{ otherwise.}
\end{cases}
\end{equation}
This matrix makes agents in a bin equally disseminated over its neighbour bins. 
In addition, 
the secondary feedback gain $G^j_k[i]$ is defined as
\begin{equation}\label{eqn:G_k}
G^j_k[i] := \left(1+\exp \Big(\gamma ( q_i - \frac{\bar{\mu}_k^j[i]}{\bar{\Theta}[i]} ) \Big) \right)^{-1},
\end{equation}
where $\gamma > 0$ is a design parameter, and $q_i > 1$ is the quorum for bin $\mathcal{B}_i$. 
The value of the gain is shown in Figure \ref{fig.G_P1P3}(b), varying depending on $\gamma$ and $q_i$. 
As ${\bar{\mu}_k^j[i]}/{\bar{\Theta}[i]}$ becomes higher than the quorum, $G^j_k[i]$ gets close to $1$ (i.e., $S^j_k[i,l]$ becomes more dominant than $P^j_k[i,l]$).  
The steepness of the function at the quorum value is regulated by $\gamma$. 

The existing quorum models in \cite{Halasz2007, Hsieh2008} require each agent to know 
$\mu_k^{\star}[i]$, which implies that the total number of agents $n^\mathcal{A}_k$ should be tracked in real time. 
It could be possible that some agents in a swarm unexpectedly become faulted by internal or external effects during a mission, 
which hinders for other alive agents from keeping track of $n^\mathcal{A}_k$ in a timely manner. 
On the contrary, this requirement is not the case for the quorum model in this subsection, and it works by using the local information available from $\mathcal{A}_{\mathcal{N}_k(i)}$.

\begin{figure}[t]
\centering
\subfloat[P1]{\includegraphics[width=0.5\linewidth]{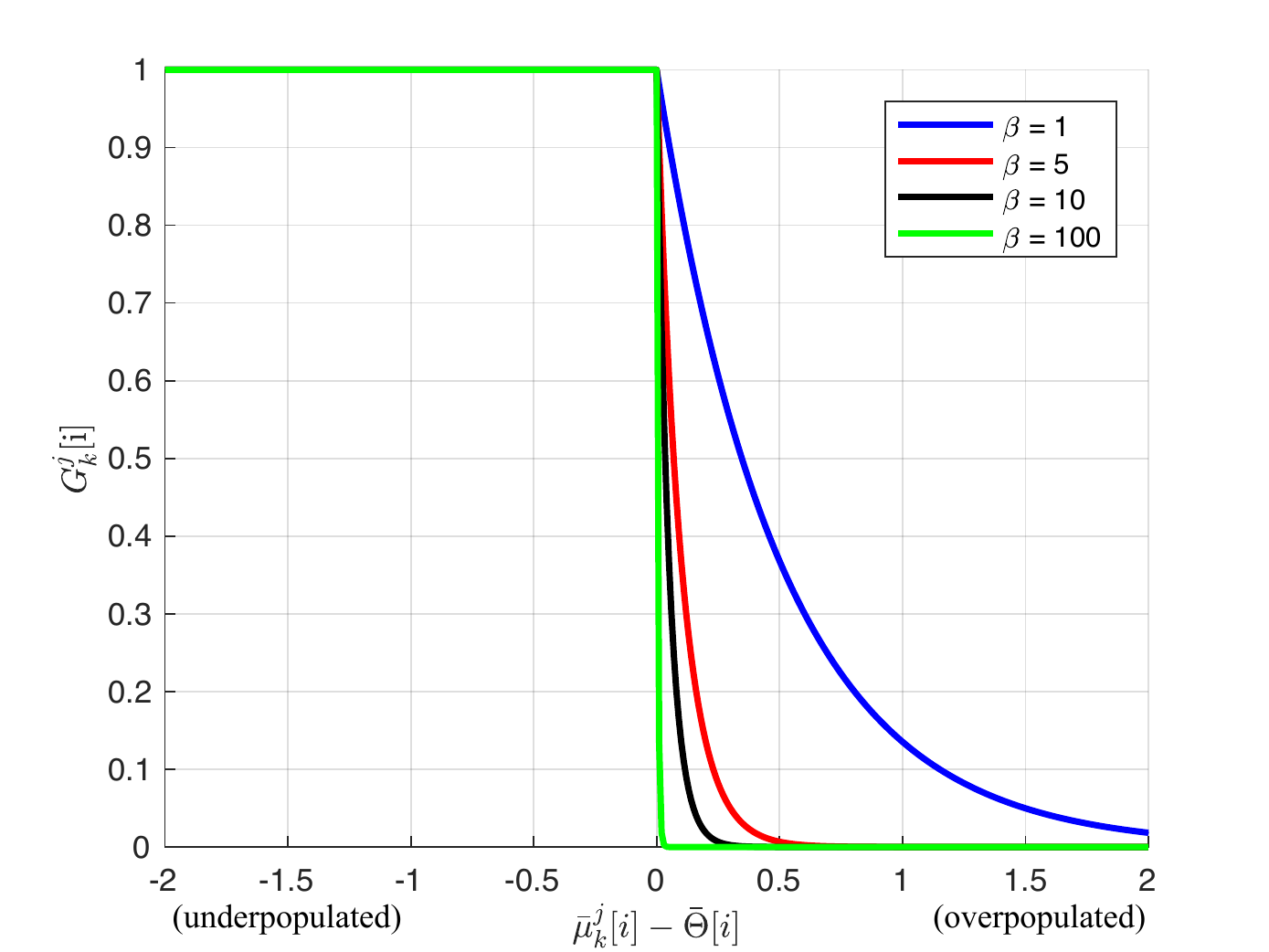}}
\hfil
\subfloat[the Quorum Model]{\includegraphics[width=0.5\linewidth]{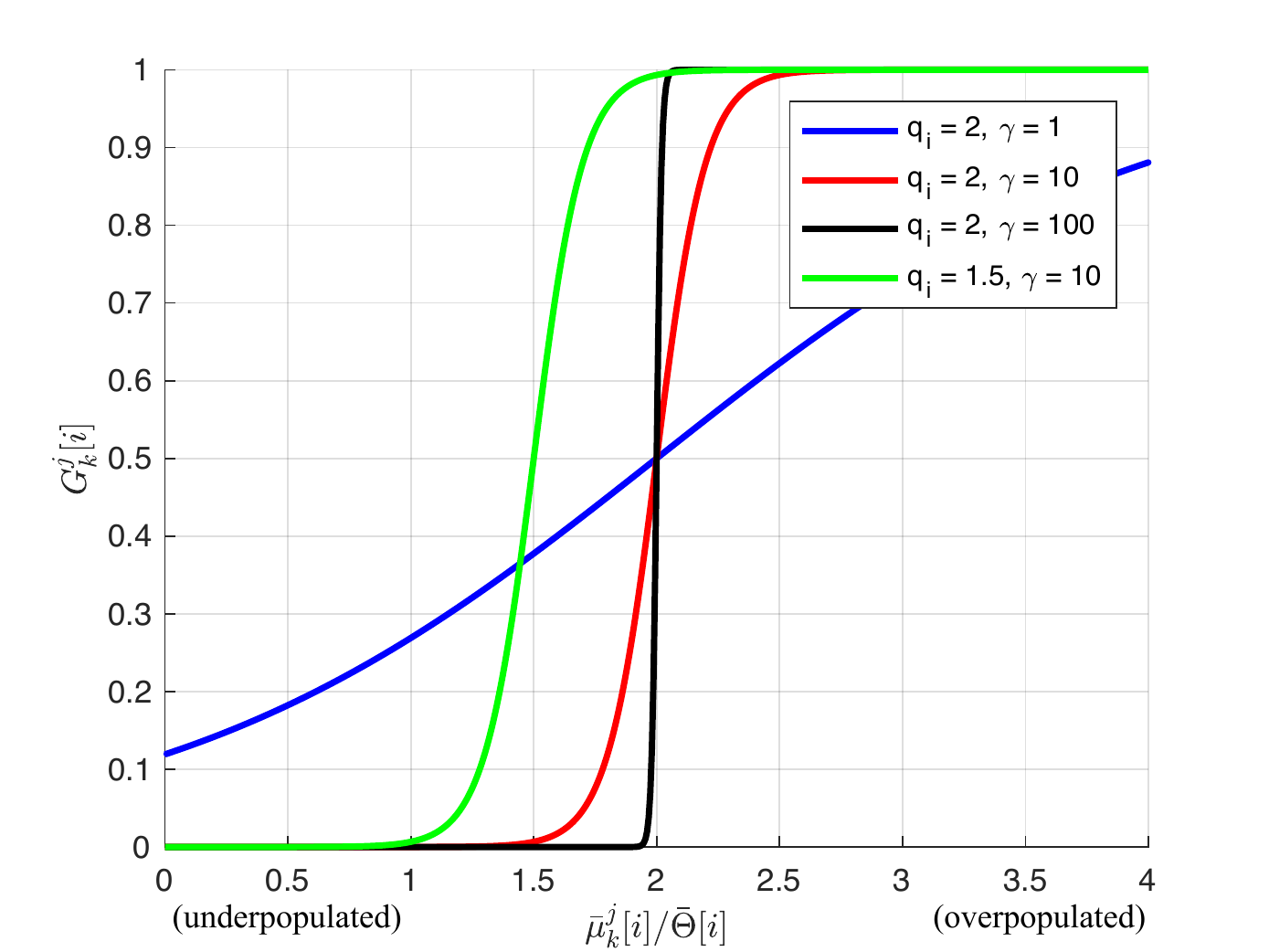}}
\hfil
\caption{The secondary feedback gains $G_k^j[i]$ depending on the associated design parameters: (1) for P1 (Eqn. (\ref{eqn:G_k_P1})); (2) for the quorum model (Eqn. (\ref{eqn:G_k}))}
\label{fig.G_P1P3}
\end{figure}

\begin{algorithm}
\caption{Subroutine of Algorithm \ref{algorithm} (Line \ref{line:Q_k}-\ref{line:G_k}) for the quorum-based method}\label{algorithm_quorum}
\begin{algorithmic}[1]

	\State Compute $S^j_k[i,l]$ $\forall l$ using (\ref{eqn:quorum_Q});	
	\State Compute $G^j_k[i]$ using (\ref{eqn:G_k});	

\end{algorithmic}
\end{algorithm}

\section{{Asynchronous Implementation}}\label{sec:async}

A synchronous process induces extra time delays and inter-agent communications to make entire agents, who may have different timescales for obtaining new information and make decisions, remain in sync. 
Such unnecessary waiting time and communications may cause unfavourable effects on mission performance or even may not be realisable in practice \cite{Johnson2011}. 

In the previous sections, it was assumed that a swarm of agents act synchronously at every time instant. 
Here we show that the proposed framework allows agents to operate in an asynchronous manner, assuming that the union of underlying graphs of the corresponding Markov matrices across some time intervals is frequently and infinitely strongly-connected. 


\begin{algorithm}[b]
\caption{{Asynchronous Construction of $P^j_k[i,l]$ (Substitute for Line \ref{line:M_k} of Algorithm \ref{algorithm})}}\label{algorithm:async}
\begin{algorithmic}[1]
\If{$\mathcal{B}_i \in \mathcal{R}_k^{+}$ \& ${\mathrm{\textbf{isnonempty}}}(\mathcal{R}_k^{+} \setminus \{\mathcal{B}_i\})$} 

		\State Compute $P^j_k[i,l]$ as usual, $\forall \mathcal{B}_l \in \mathcal{R}_k^{+} \setminus \{\mathcal{B}_i\}$; \label{algorithm:async:like_existing}
		\State $P^j_k[i,l] = 0$, $\forall \mathcal{B}_l \in \mathcal{R}_k^{-}$; \label{algorithm:async:set_zero}

\State $P^j_k[i,i] = 1 - \sum_{\forall l \neq i} P^j_k[i,l]$; \label{algorithm:async:diagonal}
	\Else
	\State $P^j_k[i,i] = 1$; $P^j_k[i,l] = 0$, $\forall l \neq i$; \label{algorithm:async_notready}
	\EndIf
\end{algorithmic}
\end{algorithm}


Suppose that an algorithm to compute $P^j_k$ that satisfies (\ref{const:stochastic})-(\ref{const:identity}) in a synchronous environment is given (e.g., Algorithm \ref{algorithm_P1} or \ref{algorithm_P2}). 
We propose an asynchronous implementation, as shown in Algorithm \ref{algorithm:async}, which substitutes for Line \ref{line:M_k} in Algorithm \ref{algorithm}. 
We refer to a set of bins where agents are ready to use their respective local information (e.g., $\bar{\mu}^j_k[i]$) as $\mathcal{R}_k^{+}$, and a set of the other bins as $\mathcal{R}_k^{-}$. 
It is assumed that each agent $j$ in bin $\mathcal{B}_i \in \mathcal{R}_k^{+}$ also knows the local information of its neighbour bin $\mathcal{B}_l \in \mathcal{N}_k(i)$ if $\mathcal{B}_l \in \mathcal{R}_k^{+}$. 
As shown in Line \ref{algorithm:async:like_existing}, the agent follows an existing procedure as long as all information required to generate $P^j_k[i,l]$ is available (e.g., $\bar{\mu}^j_k[i]$ and $\bar{\mu}^j_k[l]$ for Algorithm \ref{algorithm_P1}, and $\bar{\epsilon}_Q'[i]$ and $\bar{\epsilon}_Q'[l]$ for Algorithm \ref{algorithm_P2}).  
On the contrary, if any local information of its neighbour bin $\mathcal{B}_l \in \mathcal{N}_k(i)$ is not available, 
the probability to transition to the bin is set as zero (Line \ref{algorithm:async:set_zero}). 
In the meantime, each agent for whom necessary local information is not ready does not deviate but remains at the bin it belonged to. 
Equivalently, it can be said that $P^j_k[i,i] = 1$ and $P^j_k[i,l] = 0$, $\forall l \neq i$ (Line \ref{algorithm:async_notready}).

Hereafter, for the sake of differentiation from the original $P^j_k$ generated in a synchronous environment, let us refer to the matrix resulted by Algorithm \ref{algorithm:async} as \emph{asynchronous primary guidance matrix}, denoted by $\bar{P}^j_k$. 
Accordingly, the \emph{asynchronous Markov matrix} can be defined as:
\begin{equation*}
\bar{M}^j_k := (I - W^j_k) \bar{P}^j_k + W^j_k S^j_k.
\end{equation*}

Here, we show that this asynchronous Markov process also converges to the desired swarm distribution. 

\begin{lemma}\label{lemma:async_matrix}
The matrix $\bar{P}_{k}$, for every time instant $k$, satisfies the following properties:
(1) row-stochastic;
(2) all diagonal elements are positive, and all other elements are non-negative; 
and (3) $\sum_{i=1}^{n_{bin}} \Theta[i] \bar{P}^j_k[i,l] = \Theta[l], \forall l$.
\end{lemma}

\begin{proof}
The matrix $\bar{P}^j_{k}$ is row-stochastic because of Line \ref{algorithm:async:diagonal} and \ref{algorithm:async_notready} in Algorithm \ref{algorithm:async}. 
Furthermore, given that ${P}^j_{k}$ satisfies (\ref{const:positive_diag}), 
the property (2) is valid for $\bar{P}^j_{k}$ because $\bar{P}^j_{k}[i,i] \ge {P}^j_{k}[i,i]$ for $\forall i$.

Let us now turn to the property (3). 
For $\forall \mathcal{B}_i \in \mathcal{R}_k^{-}$, 
it is trivial that $\sum_{l=1}^{n_{bin}} \Theta[l] \bar{P}^j_k[l,i] = \Theta[i]$ because of Line \ref{algorithm:async_notready}. 
For $\forall \mathcal{B}_i \in \mathcal{R}_k^{+}$, it turns out from Algorithm \ref{algorithm:async} that (i) $\bar{P}^j_k[i,l] = \bar{P}^j_k[l,i] = 0$ for $\forall \mathcal{B}_l \in \mathcal{R}_k^{-}$;
(ii) $\bar{P}^j_k[i,l] = {P}^j_k[i,l]$ for $\forall \mathcal{B}_l \in \mathcal{R}_k^{+} \setminus \{\mathcal{B}_i\}$ 
and (iii) $\bar{P}^j_k[i,i] = {P}^j_k[i,i] + \sum_{\forall l:\mathcal{B}_l \in \mathcal{R}_k^{-}}{P}^j_k[i,l]$.
We apply the findings into the following equation:
\begin{equation}\label{eqn:for_async_proof}
\begin{split}
\sum_{l=1}^{n_{bin}} \Theta[l] \bar{P}^j_k[l,i] &= \sum_{\forall l:\mathcal{B}_l \in \mathcal{R}_k^{-}} \Theta[l] \bar{P}^j_k[l,i] \\ 
& + \sum_{\forall l:\mathcal{B}_l \in \mathcal{R}_k^{+} \setminus \{\mathcal{B}_i\}} \Theta[l] \bar{P}^j_k[l,i] + \Theta[i] \bar{P}^j_k[i,i]. 
\end{split}
\end{equation}
The first term of the right hand side becomes zero because of (i). 
Due to (ii) and the fact that $\Theta[i] {P}^j_k[i,l] = \Theta[l] {P}^j_k[l,i]$ $\forall l$, 
the second term becomes $\Theta[i] \sum_{\forall l:\mathcal{B}_l \in \mathcal{R}_k^{+} \setminus \{\mathcal{B}_i\}} {P}^j_k[i,l]$. 
The last term becomes $\Theta[i]{P}^j_k[i,i] + \Theta[i]\sum_{\forall l:\mathcal{B}_l \in \mathcal{R}_k^{-}}{P}^j_k[i,l]$ because of (iii).
Putting all of them together, Equation (\ref{eqn:for_async_proof}) is equivalent to 
$\Theta[i] \sum_{\forall l:\mathcal{B}_l \in \mathcal{R}_k^{+} \setminus \{\mathcal{B}_i\}} {P}^j_k[i,l] + \Theta[i]{P}^j_k[i,i] + \Theta[i]\sum_{\forall l:\mathcal{B}_l \in \mathcal{R}_k^{-}}{P}^j_k[i,l] = \Theta[i] \sum_{l=1}^{n_{bin}} {P}^j_k[i,l] = \Theta[i]$.
%
%
%
%
%
%
%
\end{proof}

\begin{lemma}
\label{lem:r3_dynamic}
If the union of a set of underlying graphs of $\{ \bar{P}_{k_1}, \bar{P}_{k_1+1}, ..., \bar{P}_{k_2-1} \}$ is strongly-connected, 
then the matrix product $\bar{P}_{k_1,k_2} := \bar{P}_{k_1} \bar{P}_{k_1+1} \cdots \bar{P}_{k_2-1}$ is irreducible.
\end{lemma}

\begin{proof}
Since the union of a set of underlying graphs of $\{ \bar{P}_{k_1}, \bar{P}_{k_1+1}, ..., \bar{P}_{k_2-1} \}$ is strongly-connected, 
the underlying graph of $\sum_{k = k_1}^{k_2-1} \bar{P}_k$ is also strongly-connected. 
Noting that every $\bar{P}_k$, $\forall k \in \{k_1,k_1+1,...,k_2-1\}$ is a nonnegative $n_{bin} \times n_{bin}$ matrix and its diagonal elements are positive (by Lemma \ref{lemma:async_matrix}), 
it follows from \cite[Lemma 2]{Jadbabaie2003} that $\bar{P}_{k_1,k_2} \ge \gamma \sum_{k = k_1}^{k_2-1}\bar{P}_k$, where $\gamma > 0$. 
This implies that the underlying graph of $\bar{P}_{k_1,k_2}$ is strongly-connected, and thus the matrix $\bar{P}_{k_1,k_2}$ is irreducible. 
\end{proof}

\begin{theorem}
Suppose that there exists an infinite sequence of non-overlapping time intervals $[k_i,k_{i+1})$, $i=0,1,2,...$, 
such that the union of underlying graphs of $\{ \bar{P}_{k_i}, \bar{P}_{k_i+1}, ..., \bar{P}_{k_{i+1}-1} \}$ in each interval is strongly-connected. 
Let the stochastic state of agent $j$ at time instant $k \ge k_0$, governed by the corresponding Markov process from an arbitrary state $x^j_{k_0}$, be
$x^j_{k} = {x}^j_{k_0} \bar{U}^j_{k_0,k} := {x}^j_{k_0} \bar{M}^j_{k_0} \bar{M}^j_{k_0+1} \cdots \bar{M}^j_{k-1}$. 
Then, it holds that $\lim_{k \to \infty} {x}^j_k = \Theta$ pointwise for all agents, irrespective of the initial condition. 
\end{theorem}

\begin{proof}
Thanks to Lemma \ref{lemma:async_matrix} and \ref{lem:r3_dynamic}, 
the matrix product $\bar{P}_{k_i, k_{i+1}}$ for each time interval $[k_i,k_{i+1})$ satisfies (\ref{const:stochastic})-(\ref{const:irreducible}).  
Therefore, one can prove this theorem by similarly following the proof of Theorem \ref{thm:converge}.
\end{proof}

\section{Numerical Experiments}\label{sec:experiment}

\subsection{Effects of Primary Local-feedback Gain $\bar{\xi}^j_k[i]$ }\label{sec:result.xi}


Depending on the shape of primary feedback gain $\bar{\xi}^j_k[i]$, the performance of the proposed framework changes, especially with respect to convergence rate, fraction of transitioning agents, and residual convergence error. 
Let us first investigate the effect of changes in the feedback gain using Algorithm \ref{algorithm_P1} with Equation (\ref{eqn:optimal_offdiag2}).

We consider a scenario where a set of $2,000$ agents are supposed to be distributed over an arena consisting of $10 \times 10$ bins, as depicted in Figure \ref{fig.comm_burden}. 
There are vertical and horizontal paths between adjacent bins. 
Note that the agents are allowed to move at most $3$ paths away 
within a unit time instant. 
All the agents start from a bin, which reflects the fact that they are generally deployed from a base station at the beginning of a mission. 
The desired swarm distribution $\Theta$ is uniform-randomly generated at each scenario. 
The agents are assumed to estimate necessary information correctly, e.g. $\bar{\mu}^j_k[i] = \bar{\mu}^{\star}_k[i]$.

For the rigorous validation, the performance of the proposed algorithm will be compared with that of the GICA-based algorithm \cite{Bandyopadhyay2017}. To this end, 
 $f(E_k[i,l])$ is set to be the same as the corresponding coefficient in \cite[Corollary 1]{Bandyopadhyay2017}:
\begin{equation}\label{eqn:f_E_k}
 f(E_k[i,l]) := 1 - \frac{E_k[i,l]}{E_{k,max} + \epsilon_E}
\end{equation}
where $E_{k,max}$ is the maximum element of the travelling expense matrix $E_k$, and $\epsilon_E$ is a user-design parameter. 
$E_k[i,l]$ is defined as a linear function based on the distance between bin $\mathcal{B}_i$ and $\mathcal{B}_l$: 
\begin{equation}\label{eqn:E_k}
E_k[i,l] := \epsilon_{E_1} \cdot \Delta s_{(i,l)} + \epsilon_{E_0}
\end{equation}
where $\Delta s_{(i,l)}$ is the minimum required number of paths from $\mathcal{B}_i$ to $\mathcal{B}_l$; 
$\epsilon_{E_1}$ and $\epsilon_{E_0}$ are user-design parameters. 
The agents are assumed to follow any shortest route when they transition between two bins. 
The design parameters are set as follows: $\epsilon_{E_1} = 1$ and $\epsilon_{E_0} = 0.5$ in (\ref{eqn:E_k}); 
$\epsilon_E = 0.1$ in (\ref{eqn:f_E_k}); 
$\epsilon_{\xi} = 10^{-9}$ in (\ref{eqn:xi}); 
$\epsilon_M = 1$ in (\ref{eqn:optimal_offdiag2}); 
$\beta = 1.8 \times 10^5$ in (\ref{eqn:G_k_P1}); and $\tau^j = 10^{-6}$ in (\ref{eqn:eta}).

As a performance index for the closeness between the current swarm distribution $\mu_k^{\star}$ and $\Theta$, we use \emph{Hellinger Distance}, i.e., 
\begin{equation*}\label{eqn:HD}
D_{H}(\Theta, \mu_k^{\star}) := \frac{1}{\sqrt{2}} \sqrt{ \sum_{i=1}^{n_{bin}} \left(\sqrt{\Theta[i]} - \sqrt{\mu_k^{\star}[i]} \right)^2 }, 
\end{equation*}
\emph{Hellinger Distance} is known as a ``concept of measuring similarity between two distributions'' \cite{Chung1989} and is utilised as a feedback gain in the existing work \cite{Bandyopadhyay2017}.

More importantly, to examine the effects of the shape of $\bar{\xi}^j_k[i]$, we set $\alpha$ in (\ref{eqn:xi}) as $0.2, 0.4, 0.6, 0.8, 1$ and $1.2$.  

\begin{figure}[t]
\centering
\subfloat[Primary Local-feedback Gain]{\includegraphics[width=0.45\linewidth]{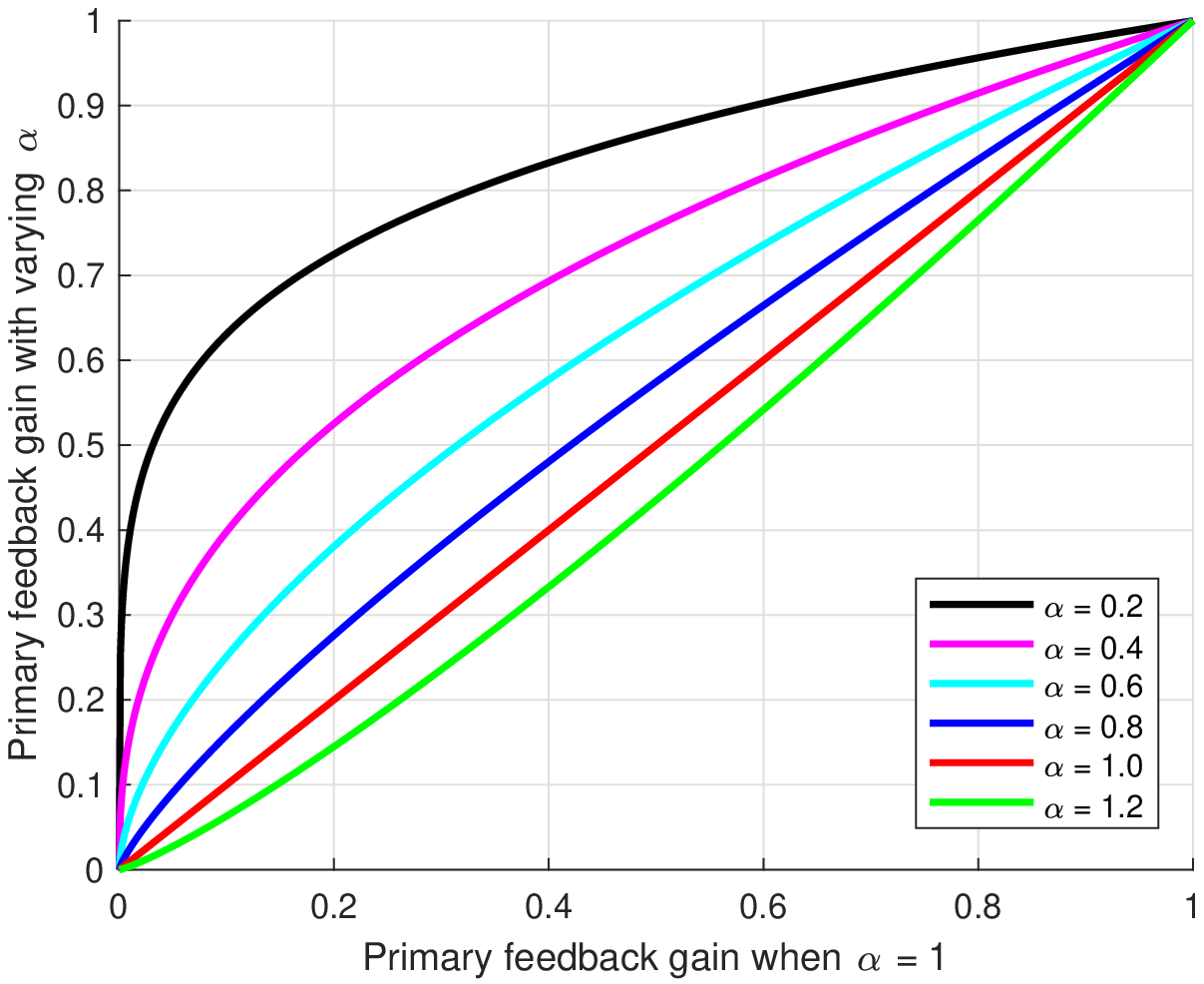}}
\hfil
\subfloat[Fraction of Transitioning Agents]{\includegraphics[width=0.45\linewidth]{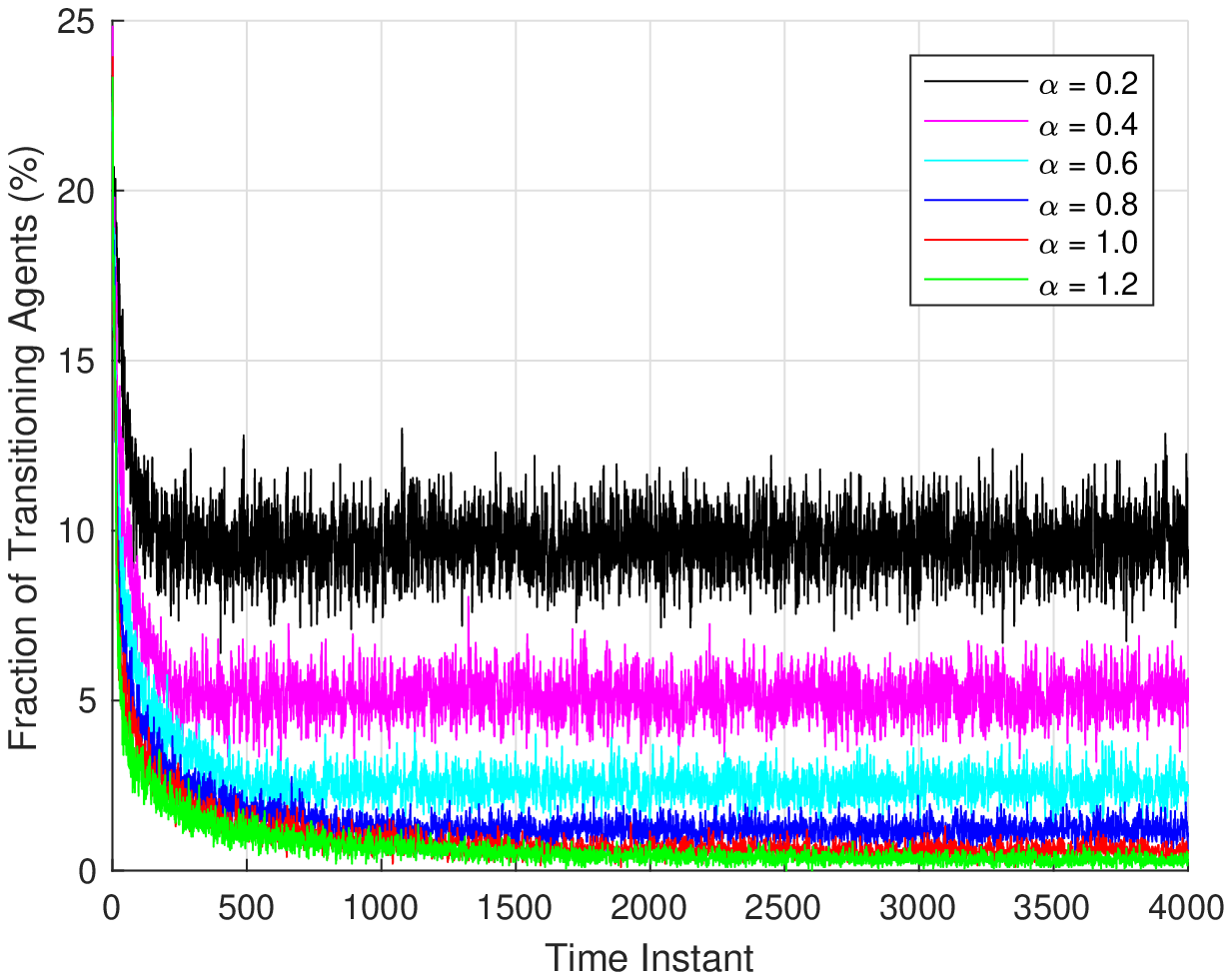}}
\hfil
\subfloat[Convergence Error]{\includegraphics[width=0.45\linewidth]{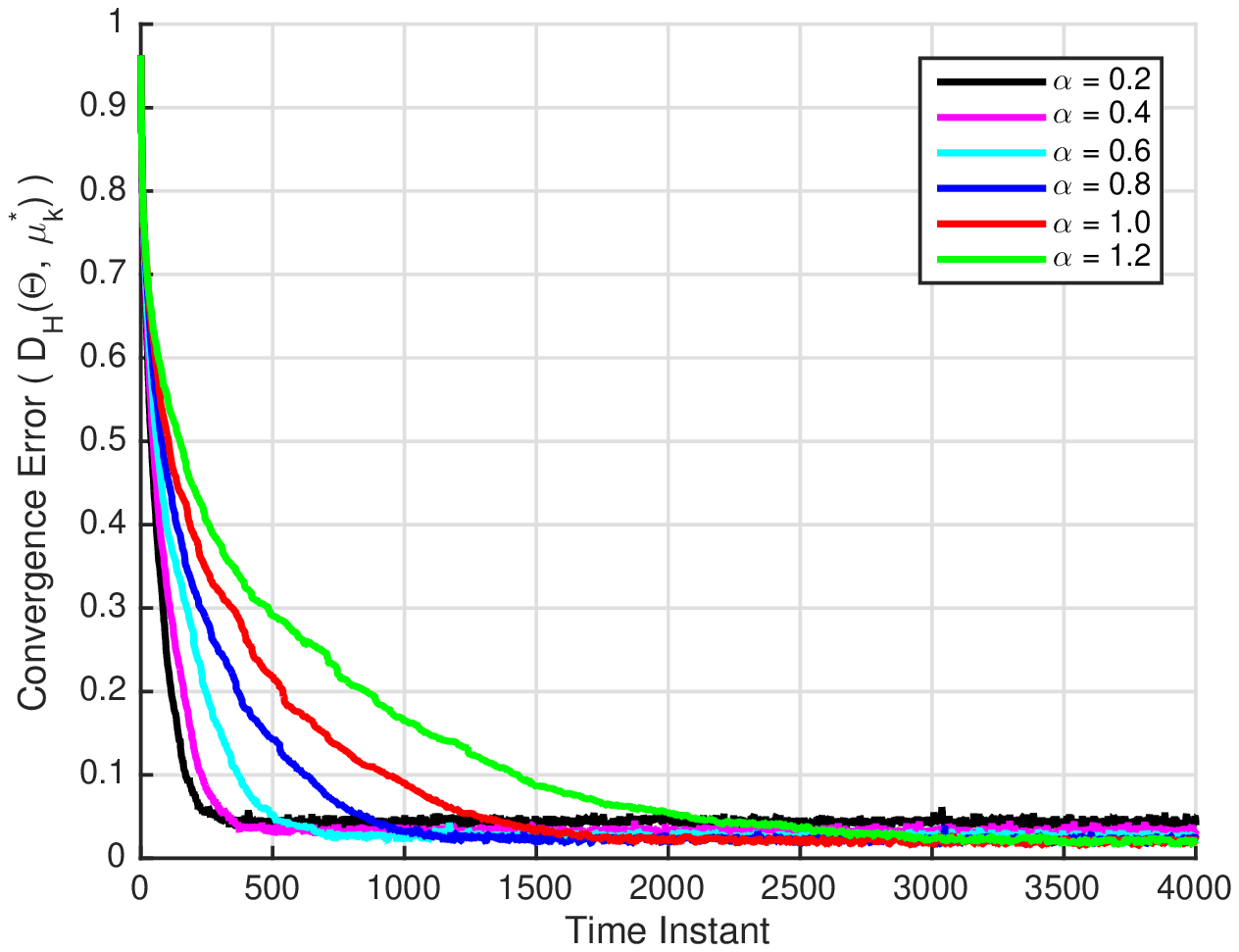}}
\hfil
\subfloat[Convergence Error (zoomed-in)]{\includegraphics[width=0.45\linewidth]{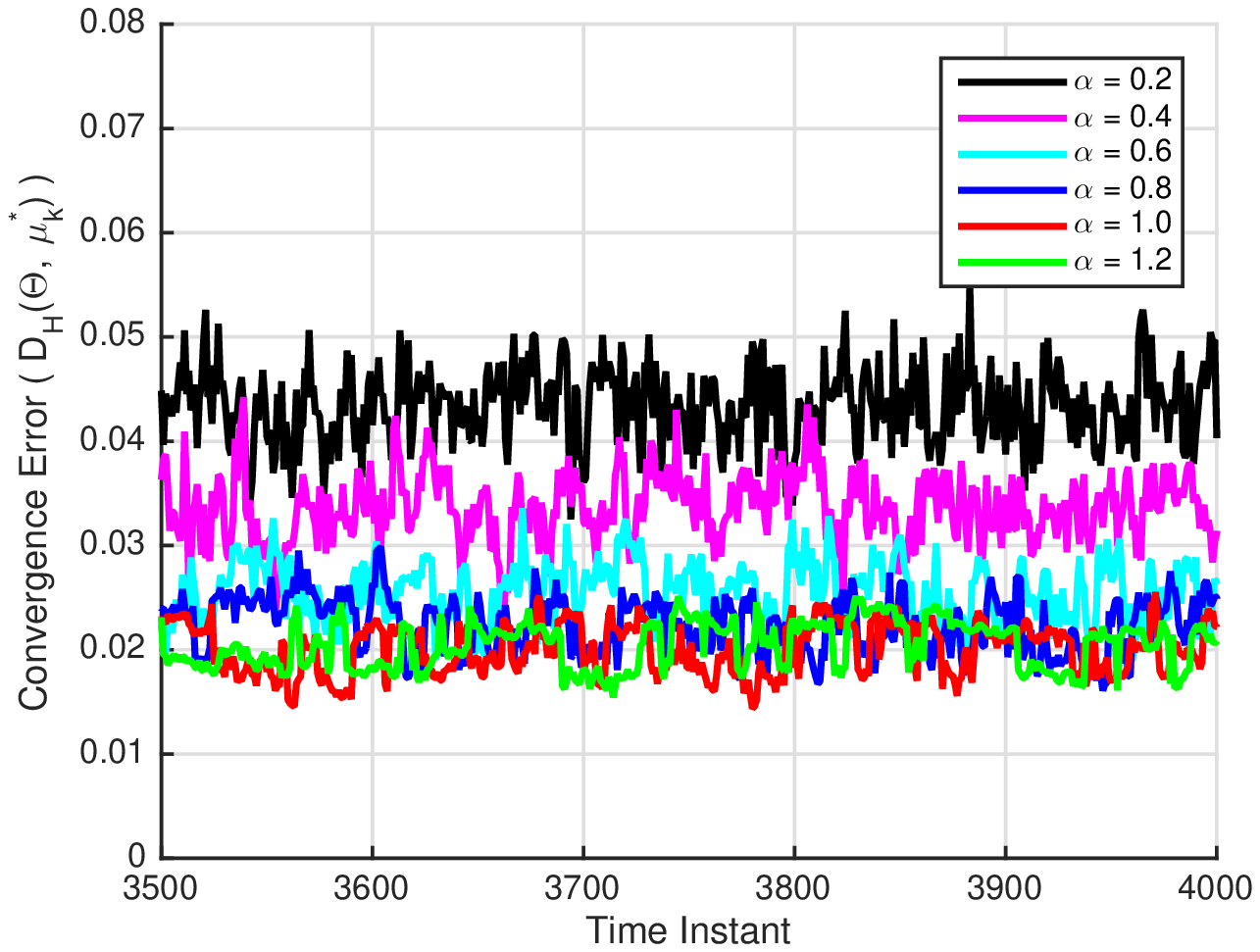}}
\hfil
\caption{Performance comparison depending on the primary local-feedback gain $\bar{\xi}_k^j[i]$ with different setting of $\alpha$: (a) the value of $\bar{\xi}_k^j[i]$; (b) the fraction of transitioning agents; (b) the convergence performance; (d) the convergence performance (zoomed-in for time instant between $2000$ and $4000$)}
\label{fig.result_xi}
\end{figure}

Figure \ref{fig.result_xi} reveals that the convergence rate can be traded off against the fraction of transitioning agents and the residual convergence error. 
As $\bar{\xi}^j_k[i]$ becomes more concave, i.e. the value of $\alpha $ decreases, the summation of off-diagonal entries of $P^j_k$ becomes higher, leading to more transitioning agents, but a faster convergence rate. 
At the same time, a higher value of $\bar{\xi}^j_k[i]$ even at a low value of $|\bar{\Theta}[i] - \mu^{j}_k[i]|$ gives rise to unnecessarily higher off-diagonal entries of $P^j_k$.  Hence, the swarm tends to be prevented from converging to the desired swarm distribution properly, resulting in higher residual convergence error.

\subsection{Comparison with a GICA-based Method}\label{sec:P1_experiment}

\begin{figure*}[htbp]
\centering
\subfloat[Convergence Error]{\includegraphics[width=0.33\linewidth]{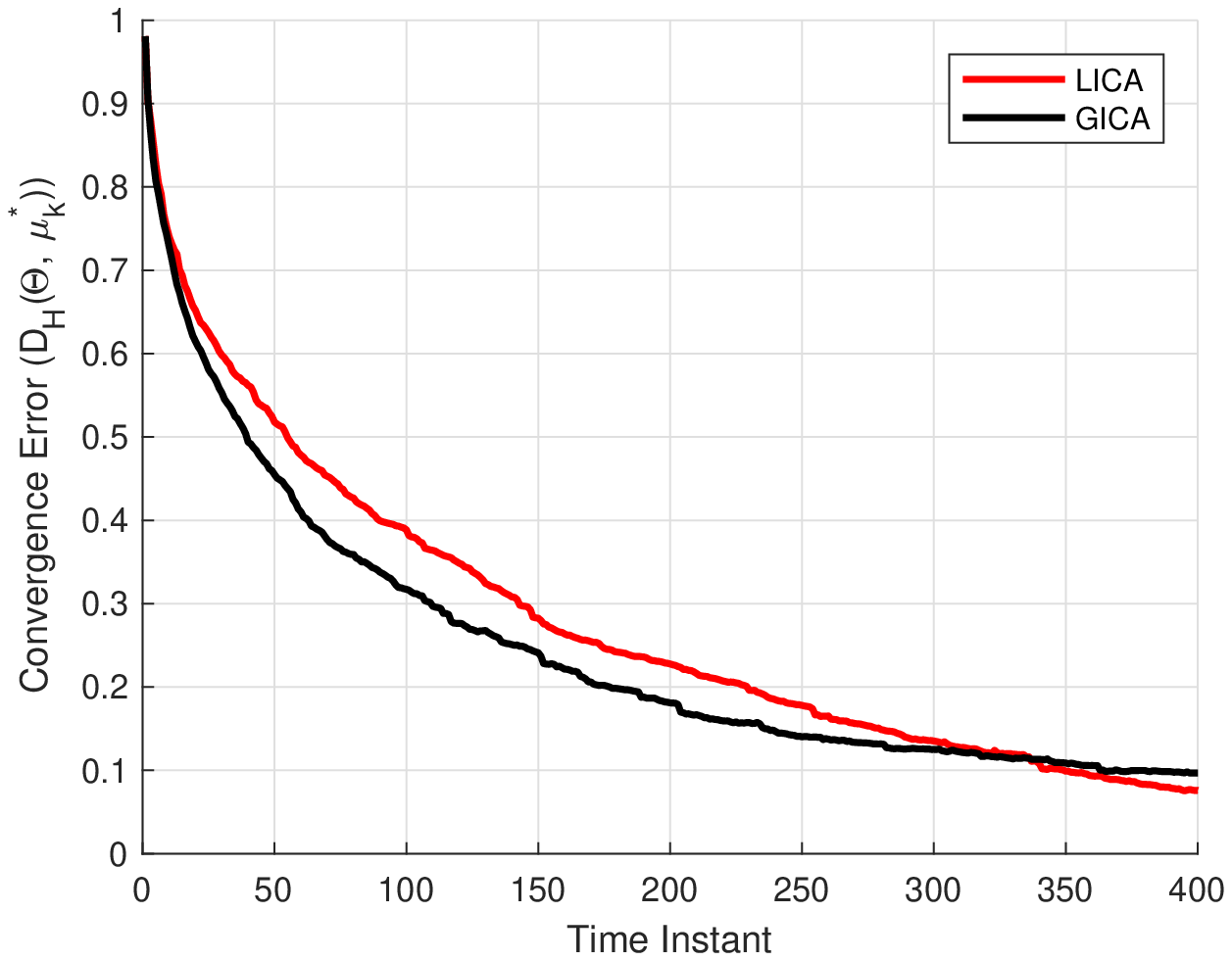}}
\hfil
\subfloat[Fraction of Transitioning Agents]{\includegraphics[width=0.33\linewidth]{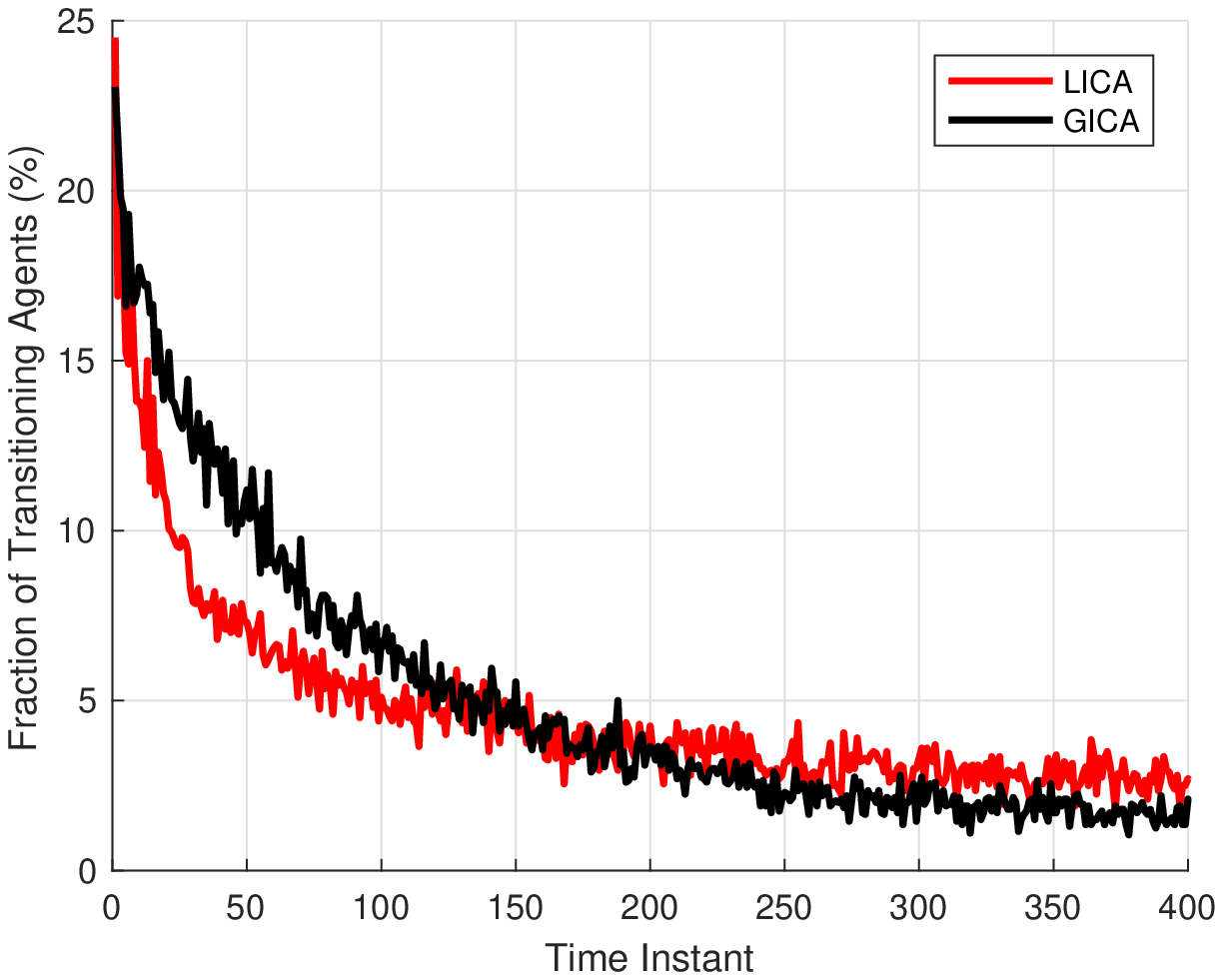}}
\hfil
\subfloat[Cumulative Travel Expenses]{\includegraphics[width=0.33\linewidth]{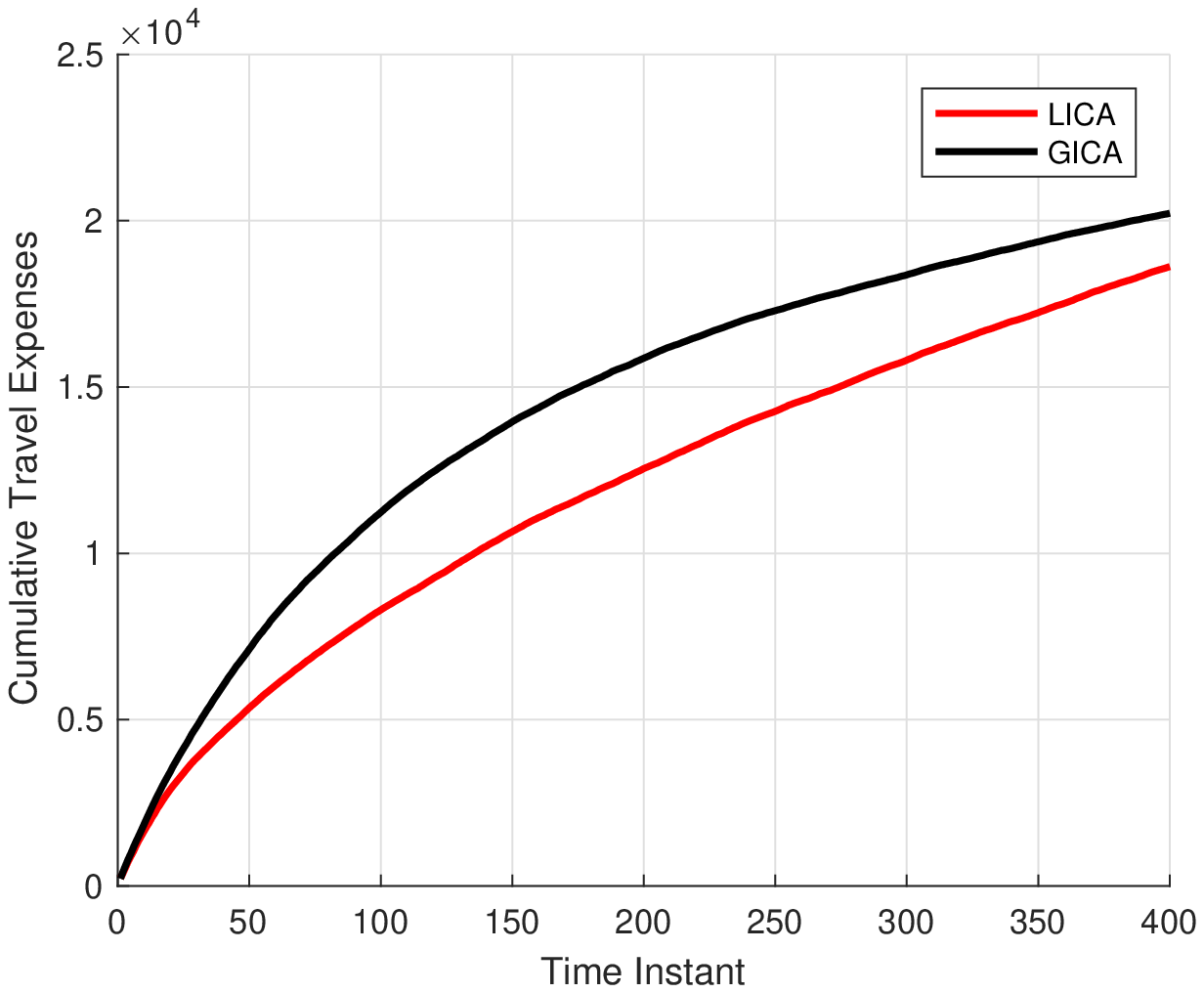}}
\hfil
\caption{Performance comparison between the proposed method (LICA) with the existing method (GICA): (a) the error between the current swarm status and the desired status, shown as Hellinger Distance; (b) the fraction of agents who transitions between any two bins; (c) the cumulative travel expenses of all the agents from the beginning.}
\label{fig.P1_comparison}
\end{figure*}

\begin{figure}[hbtp]
\centering
\subfloat[Converged Time Instants]{\includegraphics[width=0.66\linewidth]{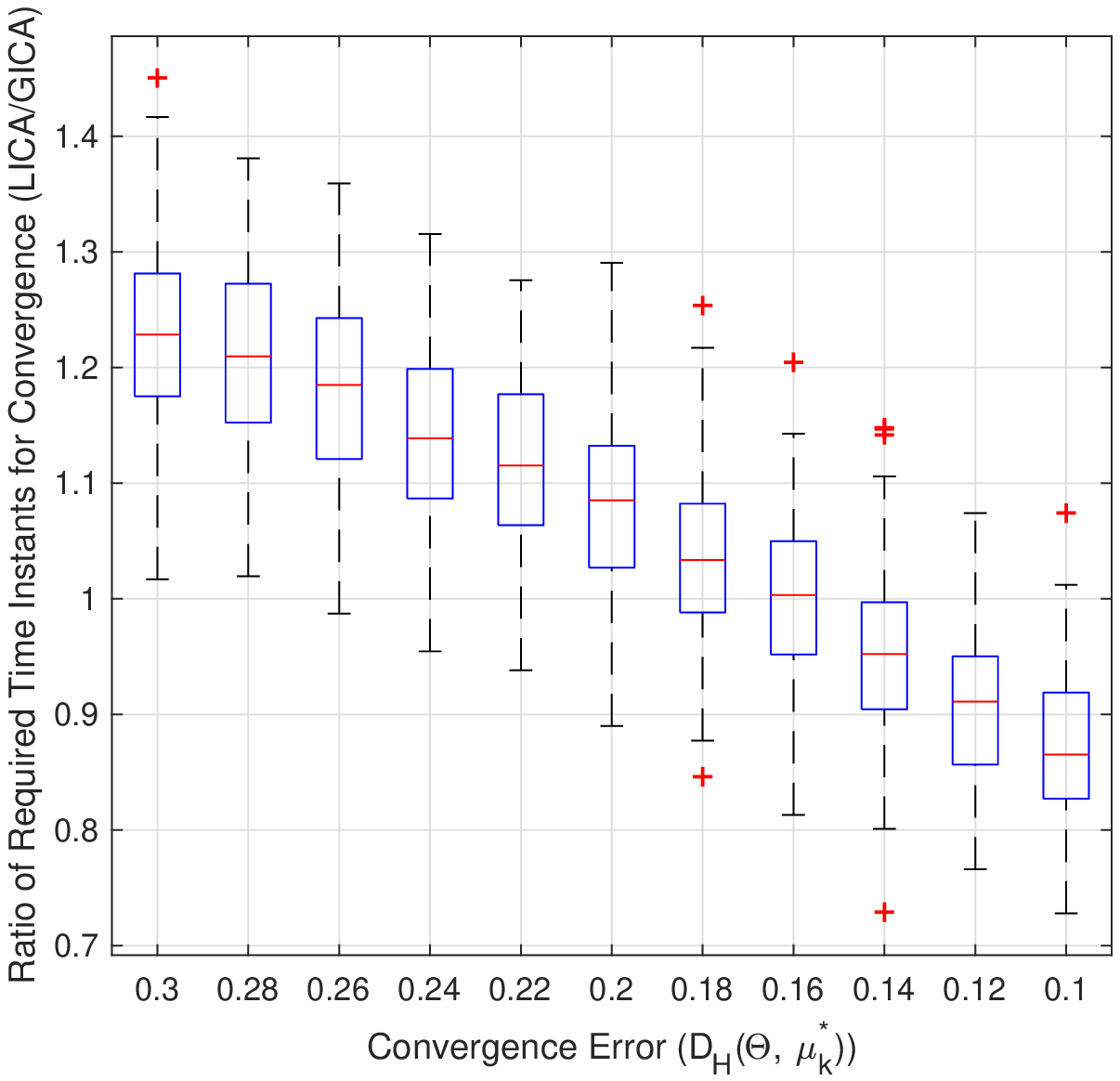}}
\hfil
\subfloat[Cumulative Travel Expenses]{\includegraphics[width=0.33\linewidth]{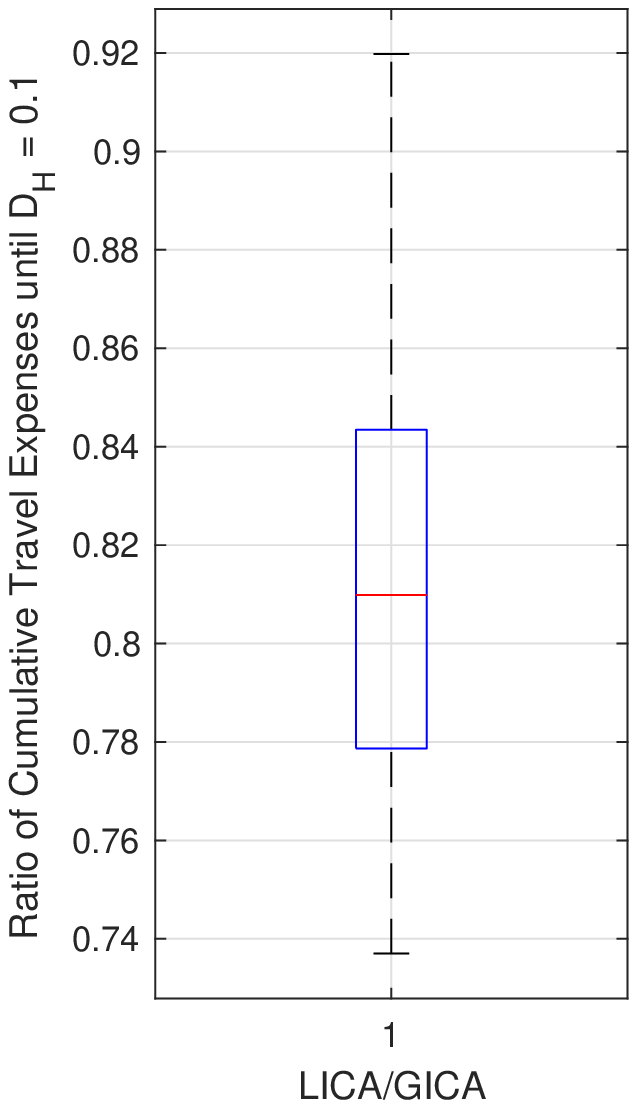}}
\hfil
\caption{Performance comparison (Monte-Carlo experiments) between the proposed method (LICA) and the existing method (GICA): (a) the required time instants to converge to $D_H(\Theta, \mu_k^{\star}) \in \{0.30, 0.28, ..., 0.12, 0.10\}$ (i.e., convergence rate); (2) the ratio of the cumulative travel expenses by LICA to those by GICA until converging to $D_H(\Theta, \mu_k^{\star})=1$.}
\label{fig.MC_comparison}
\end{figure}

Let us now compare the LICA-based method for (\ref{problem_min_cost}) with the GICA-based method in \cite{Bandyopadhyay2017}. 
The scenario considered is the same as the one in the previous subsection except for $\alpha = 0.6$. 
Note that $\epsilon_{\Theta}$ in Remark \ref{remark3} can control convergence rate, but is not discussed in \cite{Bandyopadhyay2017}. For the fair comparison,  $\epsilon_{\Theta}$ is applied to both the methods.

We conduct 100 runs of Monte Carlo experiments. 
Figure \ref{fig.P1_comparison} presents the results of one representative scenario and the statistical results of the Monte Carlo experiments are shown in Figure \ref{fig.MC_comparison}. 
According to Figure \ref{fig.P1_comparison}(a), 
the convergence rate of the proposed method is slower at the initial phase, but similar to that of the the GICA-based method as reaching $D_{H}(\Theta, \mu_k^{\star}) = 0.10$.
This is confirmed by the statistical results in Figure \ref{fig.MC_comparison}(a), 
where the ratio of the required time instants for converging to $D_{H}(\Theta, \mu_k^{\star}) \in \{0.30, 0.28, ..., 0.12, 0.10\}$ in the LICA-based method to those of the GICA-based method is presented. 
At this point, it is worth noting that these convergence rate results are presented in respect to time instants of each Markov process. 
As the LICA-based framework may have a much shorter timescale, its convergence performance in practice could be better than that of the GICA-based method.

Figure \ref{fig.P1_comparison}(c) shows that the cumulative travel expenses are smaller in the proposed method. 
The expenses by the proposed method and those by the compared method are $1.72 \times 10^4$ and $1.96 \times 10^4$, respectively, 
and their ratio is $0.878$. 
This is also confirmed by the statistical result in Figure \ref{fig.MC_comparison}(b).  
A possible explanation is that when some of the bins do not meet their desired swarm densities, the entire agents in the GICA-based method would obtain higher feedback gains, which might lead to unnecessary transitions. 
On the contrary, this is not the case in the LICA-based method since agents are only affected by their neighbour bins.


\subsection{Robustness in Asynchronous Environments}\label{sec:result.async}

\begin{figure*}[hbtp]
\centering
\subfloat[Convergence Error]{\includegraphics[width=0.33\linewidth]{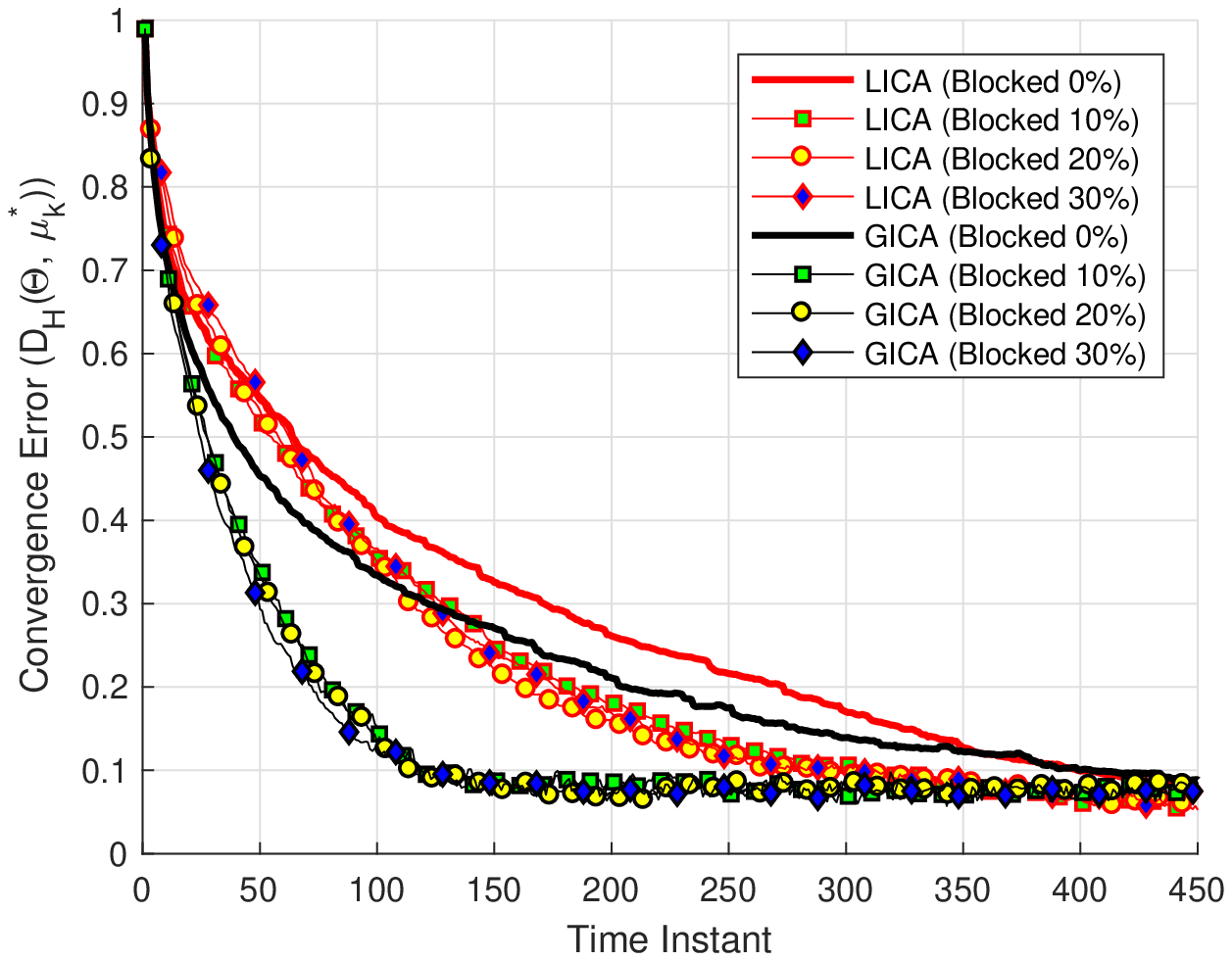}}
\hfil
\subfloat[Fraction of Transitioning Agents]{\includegraphics[width=0.33\linewidth]{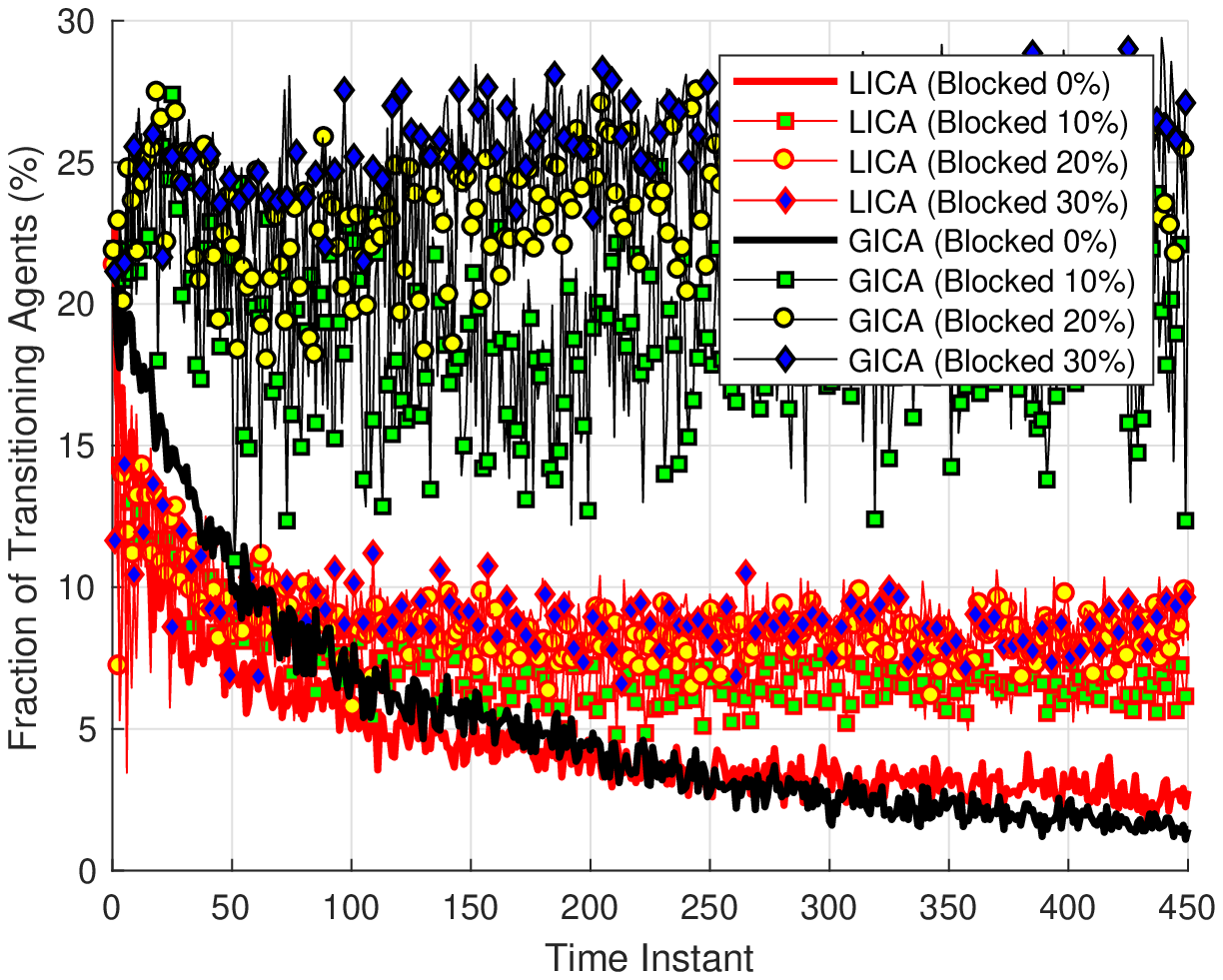}}
\hfil
\subfloat[Cumulative Travel Expenses]{\includegraphics[width=0.33\linewidth]{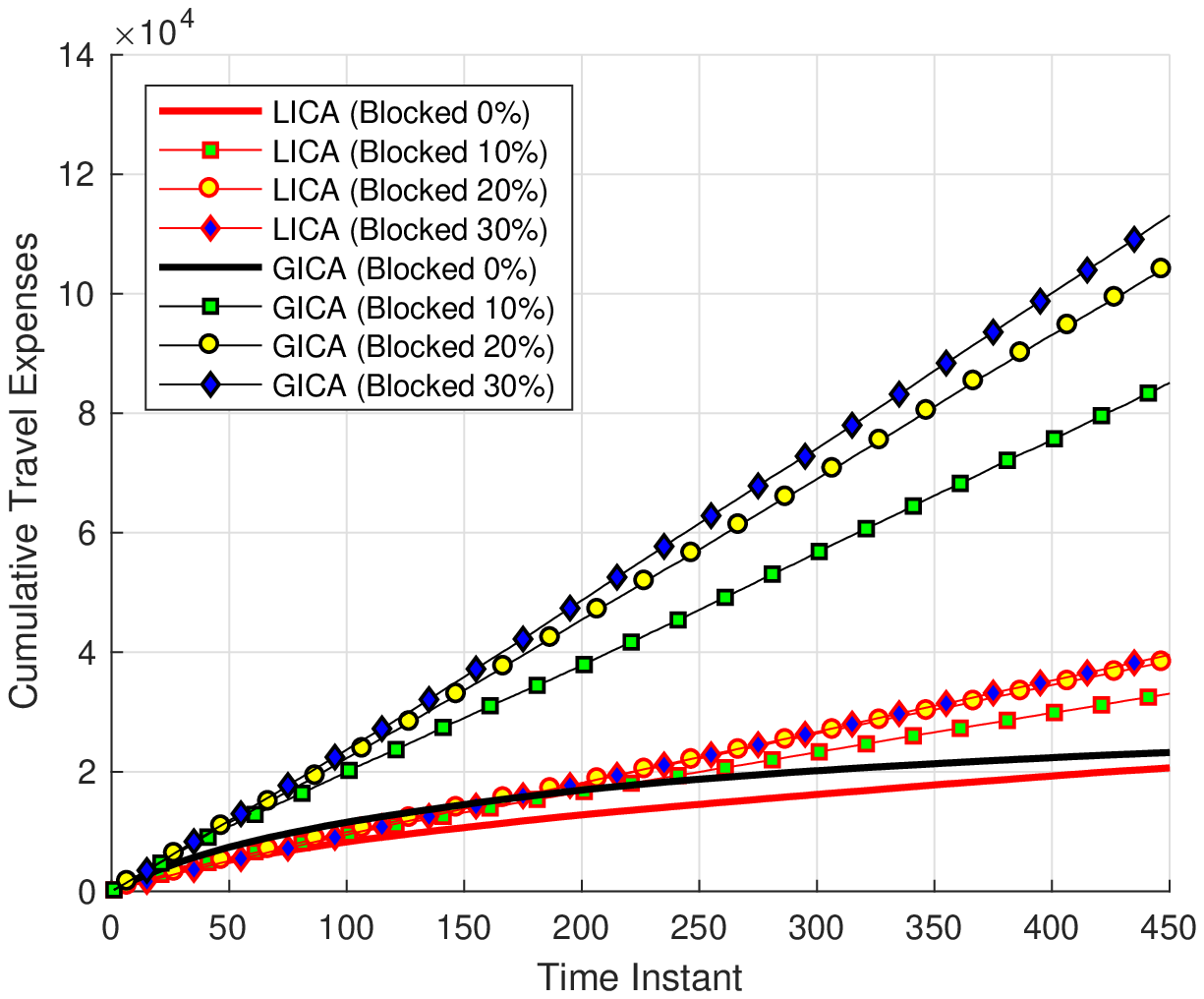}}
\hfil
\caption{Performance comparison in communication-disconnected situations: (a) the error between the current swarm status and the desired status, shown as Hellinger Distance; (b) the fraction of agents who transitions between any two bins; (c) the cumulative travel expenses of all the agents from the beginning.}
\label{fig.P_aync}
\end{figure*}
\begin{figure*}[htbp]
\centering
\subfloat[Convergence Error]{\includegraphics[width=0.33\linewidth]{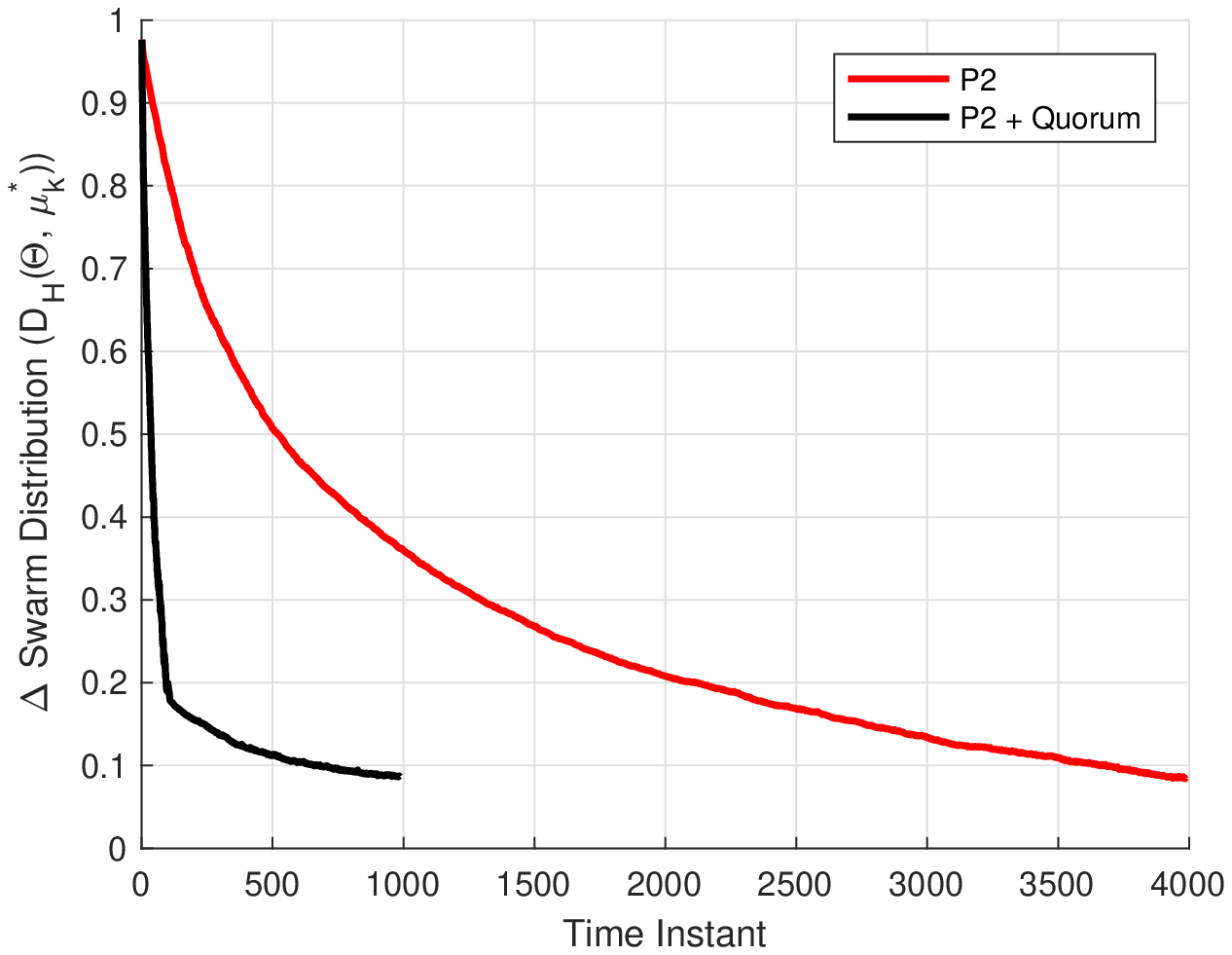}}
\hfil
\subfloat[Fraction of Transitioning Agents]{\includegraphics[width=0.33\linewidth]{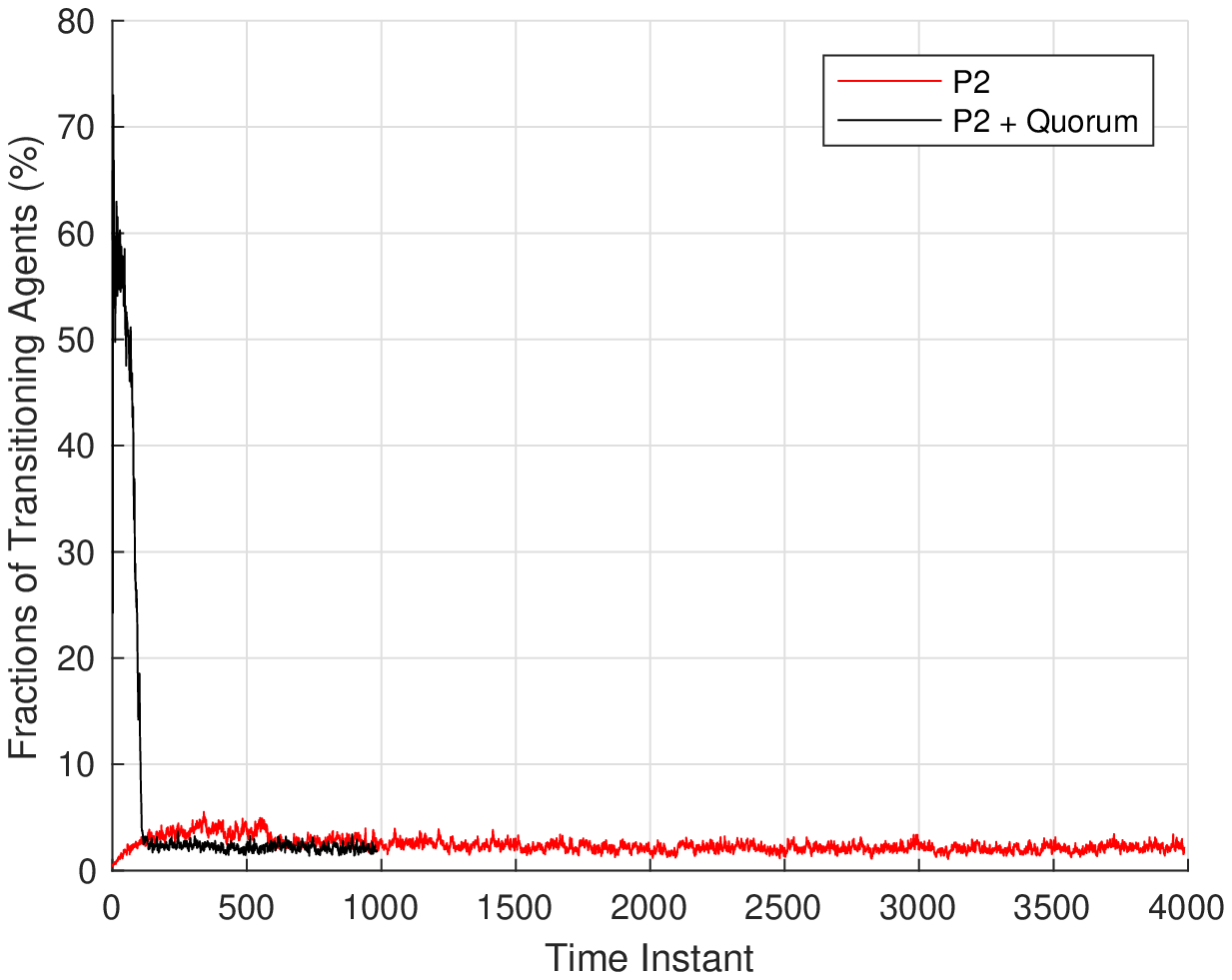}}
\hfil
\subfloat[Maximum Number of Transitioning Agents via Each Path (for P2)]{\includegraphics[width=0.33\linewidth]{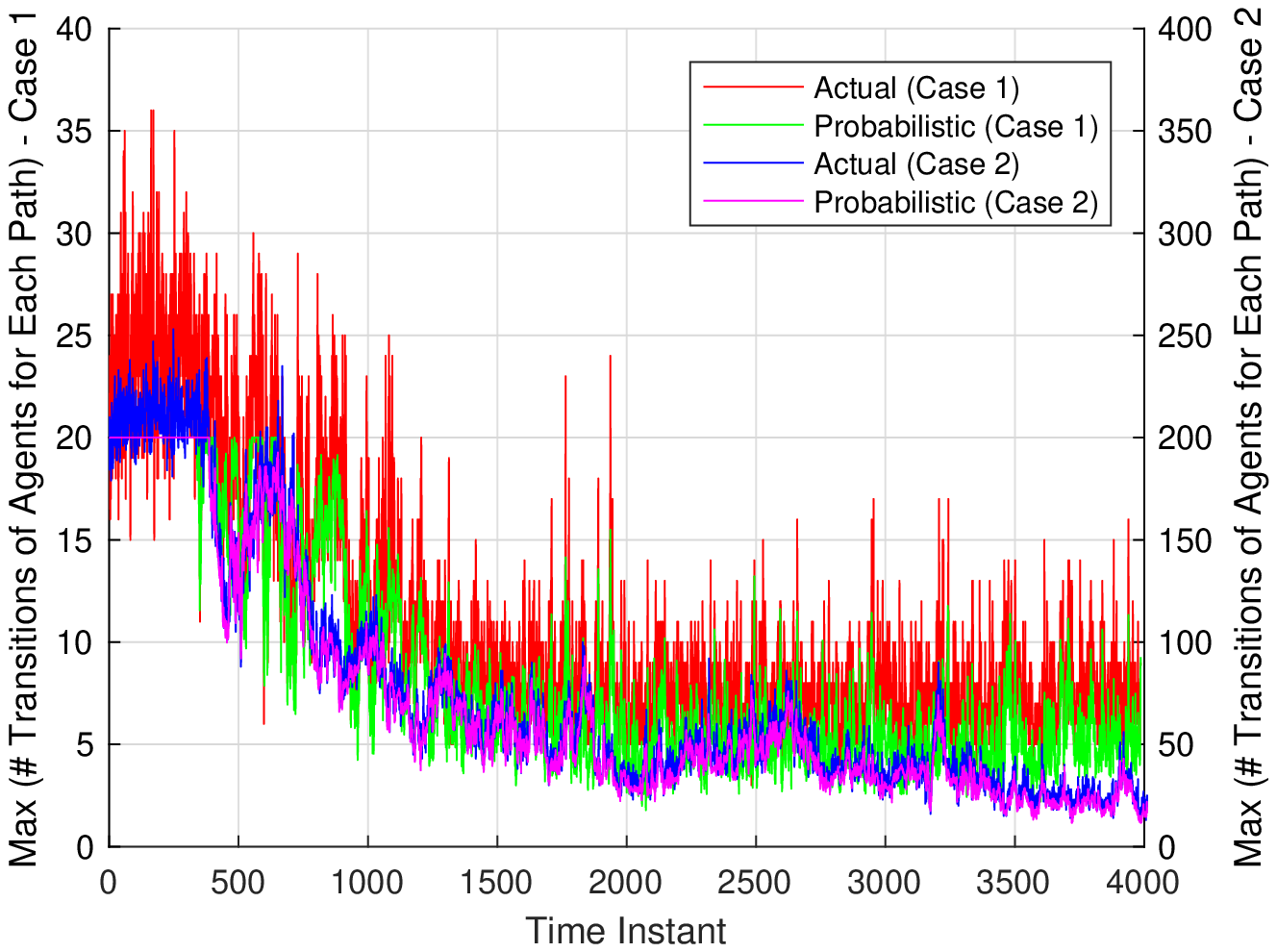}}
\hfil
\caption{Comparison results between the method for P2 and the quorum-based model: (a) the error between the current swarm status and the desired status, shown as Hellinger Distance, at each time instant; (b) the fraction of agents who transition between any two bins at each time instant. (c) The maximum number of transitioning agents via each path in the method for P2 (Case 1: $|\mathcal{A}| = 10,000$ and $c_{(i,l)} = 20$, $\forall i \neq l$; Case 2: $|\mathcal{A}| = 100,000$ and $c_{(i,l)} = 200$, $\forall i \neq l$)}
\label{fig.P2P3}
\end{figure*}

This subsection investigates the effects of asynchronous environments in the proposed LICA-based method for (\ref{problem_min_cost}) and compares them with those in the GICA-based method in \cite{Bandyopadhyay2017}. Hence, a realistic scenario where an asynchronous process is required is considered: in the scenario, it is assumed that agents in some bins cannot communicate for some reason (such bins are called \emph{blocked}) and thus other agents in normal bins have to perform their own process without waiting them. 

The proportion of blocked bins to the entire bins is set to be different values, i.e. $0\%$, $10\%$, $20\%$ and $30\%$. 
At each time instant, the corresponding proportion of bins are randomly selected as blocked bins. 
For the proposed framework, the asynchronous implementation in Section \ref{sec:async} is built upon Algorithm \ref{algorithm_P1}.
In the GICA-based method, for the comparison purpose only,  it is assumed that agents in normal bins obtain $\mu^{j}_k = \mu^{\star}_k$ without interacting with agents in the blocked bins. 
The rest of scenario setting are the same as those in Section \ref{sec:P1_experiment}. 

Figure \ref{fig.P_aync} illustrates the performance of each method: convergence rate, fraction of transitioning agents, and cumulative travel expenses. 
As the proportion of the blocked bins increases, the GICA-based method tends to have faster convergence speed, whereas it loses Desired Feature \ref{goal.no_idle_switch} and thus increases cumulative travelling expenses (as shown in Figure \ref{fig.P_aync}(a), \ref{fig.P_aync}(b), and \ref{fig.P_aync}(c), respectively). 
On the contrary, the LICA-based method shows a graceful degradation in terms of Desired Feature \ref{goal.no_idle_switch} (as shown in Figure \ref{fig.P_aync}(b)). 
A possible explanation for these results could be that higher feedback gains due to the communication disconnection induce faster convergence performance in each method than the normal situation. 
This effect is dominant for the GICA-based method because it affects the entire agents, who use global information. 
However, in the LICA-based framework, the communication disconnection only locally influences so that its effectiveness is relatively modest.

%
%
\subsection{Demonstration of Example II and III}

This subsection demonstrates the LICA-based method for (\ref{problem_max_conv}) (i.e., Algorithm \ref{algorithm_P2}) and the quorum model (i.e., Algorithm \ref{algorithm_quorum}). 
For the former, we consider a scenario where $10,000$ agents and an arena consisting of $10 \times 10$ bins are given. 
The arena is as depicted in Figure \ref{fig.comm_burden}, where the agents are allowed to move only one path away within a unit time instant. 
For each one-way path, the upper flux bound per time instant is set as $20$ agents (i.e., $c_{(i,l)} = 20$, $\forall i \neq l$). 
All the agents start from a bin, and the desired swarm distribution is uniform-randomly generated. 

For the latter, we build the quorum model upon the LICA-based method for (\ref{problem_max_conv}).
This can be a good strategy for a user who wants to achieve not only faster convergence rate but also lower unnecessary transitions after equilibrium, which are regulated by the upper flux bounds.
Thus, in the same scenario described above, we will demonstrate the combined algorithm that computes $S^j_k$ and $G^j_k$ by Algorithm \ref{algorithm_quorum} and $P^j_k$ by Algorithm \ref{algorithm_P2}. 
We set $q_i = 1.3$ and $\gamma = 30$ for (\ref{eqn:G_k}); $\alpha = 1$ and $\epsilon_{\xi} = 10^{-9}$ for (\ref{eqn:xi}); and $\tau^j = 10^{-6}$ for (\ref{eqn:eta}).

Figure \ref{fig.P2P3}(a) and \ref{fig.P2P3}(b) presents that the both approaches make the swarm converge to the desired swarm distribution. 
It is observed that the number of transitioning agents in the method for (\ref{problem_max_conv}) are restricted because of the upper flux bound during the entire process. 
Meanwhile, the quorum-based method very quickly disseminates the agents, who are initially at one bin, over other bins, and thus the fraction of transitioning agents is very high at the initial phase. After that, the population fraction by the quorum-based method drops and remains as low as that by the method for (\ref{problem_max_conv}).

Figure \ref{fig.P2P3}(c) presents the maximum value amongst the number of transitioning agents via each (one-way) path. 
The red line indicates the actual result by the method for (\ref{problem_max_conv}), while the green line indicates the corresponding probabilistic value (i.e., $\max_{\forall i}\max_{\forall l \neq i} n_k[i] P^j_k[i,l]$).
It is shown that the stochastic decision policies reflect the given upper bound, meanwhile this bound is often violated in practice due to the finite cardinality of the agents. 
However, the result in the same scenario with setting $|\mathcal{A}| = 100,000$ and $c_{(i,l)} = 200$, $\forall i \neq l$ (denoted by Case 2), depicted by the blue and magenta lines in Figure \ref{fig.P2P3}(c), suggests that such violation can be mitigated as the number of given agents increases.

\section{Conclusion}\label{Conclusion}

This paper poposed a LICA-based closed-loop-type framework for probabilistic swarm distribution guidance. 
Since the feedback gains can be generated based on local information, 
agents have shorter and different timescales for using new information, and can incorporate an asynchronous decision-making process. 
Even using local information, the proposed framework converges to a desired density distribution, while maintaining scalability, robustness, and long-term system efficiency. 
Furthermore, the numerical experiments have showed that the proposed framework is suitable for a realistic environment where communication between agents is partially and temporarily disconnected. 
This paper has explicitly presented the design requirements for the Markov matrix to hold all these advantages, and has provided specific problem examples of how to implement this framework. 

Future works include optimisation of $\bar{\xi}^j_k[i]$, which can mitigate the trade-off between convergence rate and residual error.
In addition, it is expected that the communication cost required for the proposed framework can be reduced by incorporating a vision-based local density estimation \cite{Saleh2015}.

\section*{Acknowledgement}
The authors gratefully acknowledge that this research was supported by International Joint Research Programme with Chungnam National University (No. EFA3004Z).
Thanks to Sangjun Bae for supportive discussions. 

\section*{Appendix}
\subsection{Regarding the Convergence Analysis in Theorem \ref{thm:converge}}

\begin{definition}[\emph{Irreducible}]\label{def:irreducible}
A matrix is \emph{reducible} if and only if its associated digraph is not strongly connected.
A matrix that is not reducible is \emph{irreducible}. 
\end{definition}

\begin{definition}[\emph{Primitive}] 
A \emph{primitive} matrix is a {square nonnegative} matrix $A$ such that for every $i$, $j$ there exists $k > 0$ such that $A^k[i,j] > 0$. 
\end{definition}

\begin{definition}[\emph{Regular}] 
A \emph{regular} matrix is a stochastic matrix such that all the entries of some power of the matrix are positive. 
\end{definition}

\begin{definition}\cite[pp.92, 149]{Seneta1981}\cite{Bandyopadhyay2017} (\emph{Asymptotic Homogeneity})
``A sequence of stochastic matrices $\mathcal{M}_k \in \mathbb{M}^{n \times n}$, $k \ge k_0$, is said to be \emph{asymptotically homogeneous} (with respect to $d$) 
if there exists a row-stochastic vector $d \in \mathbb{P}^{n}$ such that $\lim_{k \to \infty} d \mathcal{M}_k = d$.'' 
\end{definition}

\begin{definition}\label{def:ergodic}\cite[pp.92, 149]{Seneta1981}\cite{Bandyopadhyay2017} (\emph{Strong Ergodicity})
The matrix product $\mathcal{U}_{k_0,k} := \mathcal{M}_{k_0}\mathcal{M}_{k_0+1} \cdots \mathcal{M}_{k-1}$, formed from a sequence of stochastic matrices $\mathcal{M}_r \in \mathbb{P}^{n \times n}$, $r \ge k_0$, is said to be 
\emph{strongly ergodic} if for each $i,l,r$ we get $\lim_{r \to \infty} \mathcal{U}_{k_0,r}[i,l] = \textbf{v}[l]$, where $\textbf{v} \in \mathbb{P}^{n}$ is a row-stochastic vector. 
Here, $\textbf{v}$ is called its unique limit vector (i.e., $\lim_{r \to \infty} \mathcal{U}_{k_0,r} = \textbf{1}^{\top}\textbf{v}$). 
\end{definition}

\begin{lemma} \label{lemma:P_U} 
Given the requirements (\ref{const:stochastic})-(\ref{const:irreducible}) are satisfied, $M^j_{k}$ in Equation (\ref{eqn:markov}) has the following properties: 
\begin{enumerate}
\item row-stochastic;
\item irreducible; 
\item all diagonal elements are positive, and all other elements are non-negative;
\item there is a {positive lower bound} $\gamma$ such that $0 < \gamma \le \min_{i,l}^{+} M^j_{k}[i,l]$ (Note that $\min^+$ refers to the minimum of the positive elements);
\item asymptotically homogeneous with respect to $\Theta$. 
\end{enumerate}

In addition, $U^j_{k_0,k}$ in Equation (\ref{eqn:U}) has the following properties: 
\begin{enumerate}\addtocounter{enumi}{5}
\item irreducible; 
\item all diagonal elements are positive, and all other elements are non-negative; 
\item primitive \cite[Lemma 8.5.4, p.541]{Horn2012}  
\end{enumerate}
\end{lemma}

\begin{proof}
This lemma can be proved by similarly following the mathematical development for \cite[Theorem 4]{Bandyopadhyay2017}.
$M^j_k$ is row-stochastic because $M^j_k \cdot \textbf{1}^{\top} = (I - W^j_k) P^j_k \cdot \textbf{1}^{\top}  + W^j_k S^j_k \cdot \textbf{1}^{\top}  = (I - W^j_k) \cdot \textbf{1}^{\top}  + W^j_k \cdot \textbf{1}^{\top}  = \textbf{1}^{\top}$.
$P^j_k$ is irreducible and $\omega^j_k[i]$ is always less than 1, thus $M^j_k$ is also irreducible (i.e., $M^j_{k}[i,l] > 0$ if $P^j_{k}[i,l] > 0$). 
The property 3) is true because 
	$\mathrm{diag}(I - W^j_k) > \textbf{0}$, $W^j_k \ge \textbf{0}$, and $P^j_k$ is also a non-negative matrix such that its diagonal elements are positive. 
The property 4) is implied by either the property 2) or 3). 
From the definition of $W^j_k$, it follows that $\lim_{k \to \infty} W^j_k = \textbf{0}$ (because of $\exp(-\tau^j k)$), and thereby $\lim_{k \to \infty} M^j_k = \lim_{k \to \infty} P^j_k$. 
Hence, $\lim_{k \to \infty} \Theta M^j_k = \lim_{k \to \infty} \Theta P^j_k = \Theta$, and the property 5) is valid.

Let us now turn to $U^j_{k_0,k}$. It is irreducible due to the fact that 
if $M^j_{r}[i,l] > 0$ for some $r \in \{k_0,...,k-1\}$ and $\forall i,l \in \{1,...,n_{bin}\}$, then the corresponding element $U^j_{k_0,k}[i,l]$ is greater or equal to the product of positive diagonal elements and $M_r^j[i,l]$ and the property 2). 
The property 7) is true because $U^j_{k_0,k}$ is a product of nonnegative matrices where all diagonal entries are positive. 
It follows from \cite[Lemma 8.5.4, p.541]{Horn2012} that $U^j_{k_0,k}$ is primitive: ``if a square matrix is {irreducible}, {nonnegative} and {all its main diagonal entries are positive}, then the matrix is primitive".
\end{proof}

\ifCLASSOPTIONcaptionsoff
  \newpage
\fi



\bibliographystyle{IEEEtran}
\bibliography{library}

\end{document}